\documentclass[UKenglish]{eptcs}
\usepackage{enumerate}

\usepackage{mathtools}
\usepackage{mathpartir,amsmath,amssymb
}
\usepackage{float}
\usepackage{caption}
\usepackage{wrapfig,subfig}
\usepackage{float}
\usepackage{prooftree}
\usepackage{wrapfig,subfig}
\floatstyle{boxed}
\restylefloat{figure}   
\usepackage{listings}

\usepackage{amsmath,amssymb,wasysym,xspace}
\usepackage{mathptmx,stmaryrd}
\usepackage{eurosym}
\usepackage[T1]{fontenc}
\usepackage[mathscr]{euscript}
\usepackage[usenames,dvipsnames]{color}
\usepackage{epsfig}
\usepackage{amsfonts}
\usepackage{latexsym}
\usepackage{breakurl}
\usepackage{listings}
\usepackage{enumerate}
\usepackage{url}
\usepackage{hyperref}
\usepackage{wasysym}
\usepackage{bbm}
\newtheorem{proposition}{Proposition}[section]

\newtheorem{lemma}[proposition]{Lemma}

\newtheorem{definition}[proposition]{Definition}

\newtheorem{theorem}[proposition]{Theorem}

\newtheorem{example}[proposition]{Example}

\newcommand{\chc}{Chc}
\newcommand{\ck}{Ck}
\newcommand{\rb}{Rb}
\newcommand{\snd}{Snd}

\newcommand{\cinferrule}[3][]{
  \mprset{fraction={===},
  fractionaboveskip=0.2ex,
  fractionbelowskip=0.4ex}
  \inferrule[#1]{#2}{#3}
}

\newcommand{\tout}[3]{#1!#2(#3)}
\newcommand{\rulename}[1]{\text{\small[\textsc{#1}]}}

\def\introrule{{\cal I}}\def\elimrule{{\cal E}}

\newdimen\proofrulebreadth \proofrulebreadth=.05em
\newdimen\proofdotseparation \proofdotseparation=1.25ex
\newdimen\proofrulebaseline \proofrulebaseline=2ex
\newcount\proofdotnumber \proofdotnumber=3
\let\then\relax
\def\hfi{\hskip0pt plus.0001fil}
\mathchardef\squigto="3A3B
\newif\ifinsideprooftree\insideprooftreefalse
\newif\ifonleftofproofrule\onleftofproofrulefalse
\newif\ifproofdots\proofdotsfalse
\newif\ifdoubleproof\doubleprooffalse
\let\wereinproofbit\relax
\newdimen\shortenproofleft
\newdimen\shortenproofright
\newdimen\proofbelowshift
\newbox\proofabove
\newbox\proofbelow
\newbox\proofrulename
\def\shiftproofbelow{\let\next\relax\afterassignment\setshiftproofbelow\dimen0 }
\def\shiftproofbelowneg{\def\next{\multiply\dimen0 by-1 }%
\afterassignment\setshiftproofbelow\dimen0 }
\def\setshiftproofbelow{\next\proofbelowshift=\dimen0 }
\def\setproofrulebreadth{\proofrulebreadth}

\def\prooftree{
\ifnum  \lastpenalty=1
\then   \unpenalty
\else   \onleftofproofrulefalse
\fi
\ifonleftofproofrule
\else   \ifinsideprooftree
        \then   \hskip.5em plus1fil
        \fi
\fi
\bgroup
\setbox\proofbelow=\hbox{}\setbox\proofrulename=\hbox{}%
\let\justifies\proofover\let\leadsto\proofoverdots\let\Justifies\proofoverdbl
\let\using\proofusing\let\[\prooftree
\ifinsideprooftree\let\]\endprooftree\fi
\proofdotsfalse\doubleprooffalse
\let\thickness\setproofrulebreadth
\let\shiftright\shiftproofbelow \let\shift\shiftproofbelow
\let\shiftleft\shiftproofbelowneg
\let\ifwasinsideprooftree\ifinsideprooftree
\insideprooftreetrue
\setbox\proofabove=\hbox\bgroup$\displaystyle 
\let\wereinproofbit\prooftree
\shortenproofleft=0pt \shortenproofright=0pt \proofbelowshift=0pt
\onleftofproofruletrue\penalty1
}

\def\eproofbit{
\ifx    \wereinproofbit\prooftree
\then   \ifcase \lastpenalty
        \then   \shortenproofright=0pt  
        \or     \unpenalty\hfil         
        \or     \unpenalty\unskip       
        \else   \shortenproofright=0pt  
        \fi
\fi
\global\dimen0=\shortenproofleft
\global\dimen1=\shortenproofright
\global\dimen2=\proofrulebreadth
\global\dimen3=\proofbelowshift
\global\dimen4=\proofdotseparation
\global\count255=\proofdotnumber
$\egroup  
\shortenproofleft=\dimen0
\shortenproofright=\dimen1
\proofrulebreadth=\dimen2
\proofbelowshift=\dimen3
\proofdotseparation=\dimen4
\proofdotnumber=\count255
}

\def\proofover{
\eproofbit 
\setbox\proofbelow=\hbox\bgroup 
\let\wereinproofbit\proofover
$\displaystyle
}%
\def\proofoverdbl{
\eproofbit 
\doubleprooftrue
\setbox\proofbelow=\hbox\bgroup 
\let\wereinproofbit\proofoverdbl
$\displaystyle
}%
\def\proofoverdots{
\eproofbit 
\proofdotstrue
\setbox\proofbelow=\hbox\bgroup 
\let\wereinproofbit\proofoverdots
$\displaystyle
}%
\def\proofusing{
\eproofbit 
\setbox\proofrulename=\hbox\bgroup 
\let\wereinproofbit\proofusing
\kern0.3em$
}

\def\endprooftree{
\eproofbit 
  \dimen5 =0pt
\dimen0=\wd\proofabove \advance\dimen0-\shortenproofleft
\advance\dimen0-\shortenproofright
\dimen1=.5\dimen0 \advance\dimen1-.5\wd\proofbelow
\dimen4=\dimen1
\advance\dimen1\proofbelowshift \advance\dimen4-\proofbelowshift
\ifdim  \dimen1<0pt
\then   \advance\shortenproofleft\dimen1
        \advance\dimen0-\dimen1
        \dimen1=0pt
        \ifdim  \shortenproofleft<0pt
        \then   \setbox\proofabove=\hbox{%
                        \kern-\shortenproofleft\unhbox\proofabove}%
                \shortenproofleft=0pt
        \fi
\fi
\ifdim  \dimen4<0pt
\then   \advance\shortenproofright\dimen4
        \advance\dimen0-\dimen4
        \dimen4=0pt
\fi
\ifdim  \shortenproofright<\wd\proofrulename
\then   \shortenproofright=\wd\proofrulename
\fi
\dimen2=\shortenproofleft \advance\dimen2 by\dimen1
\dimen3=\shortenproofright\advance\dimen3 by\dimen4
\ifproofdots
\then
        \dimen6=\shortenproofleft \advance\dimen6 .5\dimen0
        \setbox1=\vbox to\proofdotseparation{\vss\hbox{$\cdot$}\vss}%
        \setbox0=\hbox{%
                \advance\dimen6-.5\wd1
                \kern\dimen6
                $\vcenter to\proofdotnumber\proofdotseparation
                        {\leaders\box1\vfill}$%
                \unhbox\proofrulename}%
\else   \dimen6=\fontdimen22\the\textfont2 
        \dimen7=\dimen6
        \advance\dimen6by.5\proofrulebreadth
        \advance\dimen7by-.5\proofrulebreadth
        \setbox0=\hbox{%
                \kern\shortenproofleft
                \ifdoubleproof
                \then   \hbox to\dimen0{%
                        $\mathsurround0pt\mathord=\mkern-6mu%
                        \cleaders\hbox{$\mkern-2mu=\mkern-2mu$}\hfill
                        \mkern-6mu\mathord=$}%
                \else   \vrule height\dimen6 depth-\dimen7 width\dimen0
                \fi
                \unhbox\proofrulename}%
        \ht0=\dimen6 \dp0=-\dimen7
\fi
\let\doll\relax
\ifwasinsideprooftree
\then   \let\VBOX\vbox
\else   \ifmmode\else$\let\doll=$\fi
        \let\VBOX\vcenter
\fi
\VBOX   {\baselineskip\proofrulebaseline \lineskip.2ex
        \expandafter\lineskiplimit\ifproofdots0ex\else-0.6ex\fi
        \hbox   spread\dimen5   {\hfi\unhbox\proofabove\hfi}%
        \hbox{\box0}%
        \hbox   {\kern\dimen2 \box\proofbelow}}\doll%
\global\dimen2=\dimen2
\global\dimen3=\dimen3
\egroup 
\ifonleftofproofrule
\then   \shortenproofleft=\dimen2
\fi
\shortenproofright=\dimen3
\onleftofproofrulefalse
\ifinsideprooftree
\then   \hskip.5em plus 1fil \penalty2
\fi
}

\newcommand{\myrule}[3]{\begin{prooftree}
 #1 \justifies   #2   \using{\rln{#3}} \end{prooftree}}
 \newcommand{\myformula}[1]{\\\centerline{#1}}

  \newenvironment{myitemize}
               {\begin{itemize}\vspace{-2pt}\topsep0pt\parskip0pt\partopsep0pt\itemsep0pt\leftmargin-100pt\itemsep-1pt\labelwidth0pt\labelsep3pt}
               {\vspace{-1pt}\end{itemize}}

 \newenvironment{myenumerate}[1]
               {\begin{enumerate}{#1}\vspace{-2pt}\topsep0pt\parskip0pt\partopsep0pt\itemsep0pt\leftmargin10pt\itemsep-2pt\labelwidth0pt\labelsep3pt}
               {\vspace{-2pt}\end{enumerate}}

               \newenvironment{mylemma}
               {\vspace{-1pt}
               \begin{lemma}\vspace{-1pt}
               }
               {\end{lemma}}

\newcommand{\A}{{\mathcal A}}

\newcommand{\pc}{~|~}

\newcommand{\set}[1]{\{#1\}}

\newcommand{\labelx}[1]{\label{#1}}

\newcommand{\s}{\epsilon}

\newcommand{\Gvti}[4]{\ensuremath{#1\to\pset:\{#2_i({#3}_i). #4_i\}_{i \in I}}}
\newcommand{\Gvtj}[4]{\ensuremath{#1\to\pset:\{#2_j({#3}_j). #4_j\}_{j \in J}}} 
\newcommand{\Gvth}[4]{\ensuremath{#1\to\pset:\{#2_h({#3}_h). #4_h\}_{h \in H}}}   
\newcommand{\GvtiLong}[3]{\ensuremath{#1\to#2:\{\ #3\ \}}}
\newcommand{\pp}{{\sf p}}
\newcommand{\q}{{\sf q}}
\newcommand{\pr}{{\sf r}}

\newcommand{\G}{\ensuremath{{\sf G}}}
\newcommand{\ST}{\ensuremath{S}}

\newcommand{\T}{\ensuremath{\mathsf{T}}}

\renewcommand{\P}{\ensuremath{P}}

\newcommand{\pset}{\ensuremath{\q}}

\newcommand{\ty}{\textbf{t}}

\newcommand{\GG}{\Upsilon}

\newcommand{\End}{\kf{end}}
\newcommand{\next}{\kf{next}}

\newcommand{\inact}{\ensuremath{\mathbf{0}}}
\newcommand{\kf}[1]{\ensuremath{\mathsf{#1}\xspace}}

\newcommand{\e}{\kf{e}}
\newcommand{\x}{x}
\newcommand{\val}{v}
\newcommand{\eval}[2]{#1 \downarrow #2}

\newcommand{\recvar}{\ensuremath{X}}

\newcommand{\sep}{\ensuremath{~\mathbf{|\!\!|}~ }}

\newcommand{\stackred}[1]{\xrightarrow{#1}}
\newcommand{\stackredstar}[1]{\xrightarrow{#1}^{\raisebox{-4pt}{$\ast$}}}
\newcommand{\proj}[2]{#1 \!  \upharpoonright  \! #2\,}

\newcommand{\Q}{\ensuremath{Q}}

\newcommand{\pt}{\kf{pt}}
\newcommand{\rln}[1]{{\textsc{#1}}}

\DeclareMathAlphabet{\mathpzc}{OT1}{pzc}{m}{it}

\definecolor{ocre}{rgb}{.92,.86,.3}%
\definecolor{verde}{rgb}{0.50,0.70,0.4}
\definecolor{beppeblue}{rgb}{0,0,0.4}
\definecolor{lucagreen}{rgb}{0,0.3,0}
\definecolor{lightgreen}{rgb}{0.79,0.86,.79}
\definecolor{lesslightgreen}{rgb}{0.39,.9,.39}
\definecolor{mydarkgreen}{rgb}{0.10,0.43,0}
\definecolor{darkred}{rgb}{0.67,0.27,0.39}
\definecolor{ultralightgray}{gray}{0.85}
\definecolor{midgray}{gray}{0.6}
\definecolor{siena}{rgb}{0.58,0.23,0.34}
\definecolor{mgreen}{rgb}{0,0.5,0}
\definecolor{lucared}{rgb}{0.8,0,0}

\newcommand{\pskip}[1]{}

\newcommand{\sendMsg}[4]{#1!#2(#3).#4}
\newcommand{\rcvMsg}[4]{#1?#2(#3).#4}

\newcommand{\sendPrcS}[2]{#1!#2}
\newcommand{\rcvPrcS}[2]{#1?#2}

\newcommand{\sendMsgS}[3]{#1!#2.#3}
\newcommand{\rcvMsgS}[3]{#1?#2.#3}

\newcommand{\GlComm}[2]{#1\to#2:}
\newcommand{\CkGlComm}[3]{{}_{\ckP{#3}}#1\to#2:}
\newcommand{\ckP}[1]{\blacktriangle_{#1}}
\newcommand{\CkOne}[2]{\ensuremath{{}_{\ckP{#2}}\!#1}}
\newcommand{\CkTwo}[2]{\ensuremath{{}_{\ckP{#2}}\!\!#1}}
\newcommand{\CkZero}[2]{\ensuremath{{}_{\ckP{#2}}#1}}

\newcommand{\CkGl}[2]{\CkZero{#1}{#2}}
\newcommand{\interTyI}[3]{\ensuremath{\bigwedge_{#1 \in #2}#3}}
\newcommand{\unionTyI}[3]{\ensuremath{\bigvee_{#1 \in #2}#3}}
\newcommand{\unionTy}[1]{\bigvee\{#1\}}
\newcommand{\interTy}[1]{\bigwedge\{#1\}}

\newcommand{\interSymD}{\bigwedge\!\!\!\!\bigwedge}

\newcommand{\interRcv}[1]{\interSymD #1}

\newcommand{\intMsg}[6]{\interTyI{#1}{#2}{\rcvMsg{#3}{#4_{#1}}{#5_{#1}}{#6_{#1}}}}
\newcommand{\unMsg}[6]{\unionTyI{#1}{#2}{\sendMsg{#3}{#4_{#1}}{#5_{#1}}{#6_{#1}}}}
\newcommand{\intMsgLong}[6]{\interTyI{#1}{#2}{\rcvMsg{#3}{#4_{#1}}{#5_{#1}}{#6}}}
\newcommand{\unMsgLong}[6]{\unionTyI{#1}{#2}{\sendMsg{#3}{#4_{#1}}{#5_{#1}}{#6}}}
\newcommand{\intChoice}[6]{\bigoplus_{#1\in #2}\sendMsg{#3}{#4_{#1}}{#5_{#1}}{#6_{#1}}}
\newcommand{\extChoice}[6]{\sum_{#1\in #2}\rcvMsg{#3}{#4_{#1}}{#5_{#1}}{#6_{#1}}}

\newcommand{\intChPr}[1]{\intChSy\{#1\}}
\newcommand{\extChPr}[1]{\extChSy\{#1\}}
\newcommand{\intChSy}{\bigoplus}
\newcommand{\extChSy}{\sum}

\newcommand{\confEl}[2]{#1\mathrel{\triangleleft}#2}

\newcommand{\BP}{R}
\newcommand{\noBP}{\epsilon}
\newcommand{\PC}{\|}
\newcommand{\pairPr}[2]{#1\prec #2}
\newcommand{\pairGT}[2]{#1\prec #2}
\newcommand{\pairT}[2]{#1\prec #2}

\newcommand{\myruleOneLine}[2]{#1\quad\rln{#2}}
\newcommand{\PrConf}{\mathbbmss{C}}
\newcommand{\Nw}{\mathbbmss{M}}
\newcommand{\Nt}{\mathbbmss{N}}
\newcommand{\subst}[3]{#1\{#3/#2\}}

\newcommand{\silent}{\tau}

\newcommand{\dere}[3]{#1\vdash #2:#3}
\newcommand{\derP}[3]{#1 \vdash #2 :#3}
\newcommand{\derN}[2]{\vdash #1 :#2}
\newcommand{\derS}[2]{\vdash #1 :#2}
\newcommand{\derNok}[1]{\vdash #1 ~\checkmark}
\newcommand{\subt}{\leqslant}

\newcommand{\Push}[2]{#1\cdot#2}

\newcommand{\PushT}[2]{#1\cdot#2}
\newcommand{\len}[1]{|#1|}
\newcommand{\RT}{\rho}

\newcommand{\redG}{\Longrightarrow}
\newcommand{\redGstar}{\Longrightarrow^{\raisebox{-1pt}{$\ast$}}}
\newcommand{\agr}[3]{#1\ltimes_#2#3}

\newenvironment{proof}[1][Proof]{\begin{trivlist}
\item[\hskip \labelsep {\bfseries #1}]}{\end{trivlist}}

\newcommand{\tbool}{\mathtt{Bool}}
\newcommand{\tint}{\mathtt{Int}}

\newcommand{\query}{\ensuremath{\mathtt{query}}}
\newcommand{\info}{\ensuremath{\mathtt{info}}}
\newcommand{\available}{\ensuremath{\mathtt{available}}}
\newcommand{\notAvailable}{\ensuremath{\mathtt{notAvailable}}}
\newcommand{\reserve}{\ensuremath{\mathtt{reserve}}}
\newcommand{\discard}{\ensuremath{\mathtt{discard}}}

\newcommand{\queryM}{\ensuremath{\mathtt{qr}}}
\newcommand{\infoM}{\ensuremath{\mathtt{in}}}
\newcommand{\availableM}{\ensuremath{\mathtt{av}}}
\newcommand{\notAvailableM}{\ensuremath{\mathtt{nAv}}}
\newcommand{\reserveM}{\ensuremath{\mathtt{rs}}}
\newcommand{\discardM}{\ensuremath{\mathtt{ds}}}
\newcommand{\stringT}{\ensuremath{\kf{Str}}}

\newcommand{\Traveler}{\ensuremath{\mathtt{Tr}}}
\newcommand{\Hotel}{\ensuremath{\mathtt{Ht}}}
\newcommand{\Airline}{\ensuremath{\mathtt{Al}}}

\newcommand{\mysubsection}[1]
{\paragraph{#1}}
\newcommand{\mysection}[1]{\vspace{-10.5pt}\section{#1}\vspace{-6.5pt}}
\newcommand{\mysubsectionA}[1]{\paragraph{#1}}
\title{Reversible Multiparty Sessions with Checkpoints\thanks{Partially supported by EU H2020-644235 Rephrase project, EU H2020-644298 HyVar project, ICT COST Actions IC1201 BETTY, IC1402 ARVI and Ateneo/CSP project RunVar.} {\footnotesize }
\vspace{-15pt}
}
\author{ 
Mariangiola Dezani-Ciancaglini
\footnote{Dipartimento di Informatica,
Universit\`a di Torino,
dezani@di.unito.it}
\and
Paola Giannini\footnote{Computer Science Institute, DiSIT,
Universit\`a del
Piemonte Orientale,
paola.giannini@uniupo.it}
}

\def\titlerunning{Reversible Multiparty Sessions with Checkpoints}
\def\authorrunning{Dezani-Ciancaglini \& Giannini}

\begin{document}

\maketitle

\vspace{-10pt}
\begin{abstract}
Reversible interactions model different scenarios, like biochemical systems and human as well as automatic negotiations. 
We abstract interactions via multiparty sessions enriched with named checkpoints. 
Computations can either go forward or roll back to some checkpoints, where possibly different choices may be taken. 
In this way communications can be undone and different conversations may be tried. 
Interactions are typed with global types, which control also rollbacks. 
Typeability of session participants in agreement with global types ensures session fidelity and progress of reversible communications.
\end{abstract}
\vspace{-10pt}

\mysection{Introduction}

{\em Reversibility} is an essential feature in the construction of reliable systems.
If a system reaches an undesired state, some actions may be undone and the computation
may be restarted from a consistent state. Several studies,
see~\cite{DK04,LMSS11,LMS10,PU07}, have investigated the theoretical foundations of reversible computations. The relevance of these papers for our work is briefly discussed at the beginning of Section~\ref{rwc}.

\begin{wrapfigure}{r}{0.35\textwidth}
\hspace*{\fill}
\includegraphics[width=0.35\textwidth]{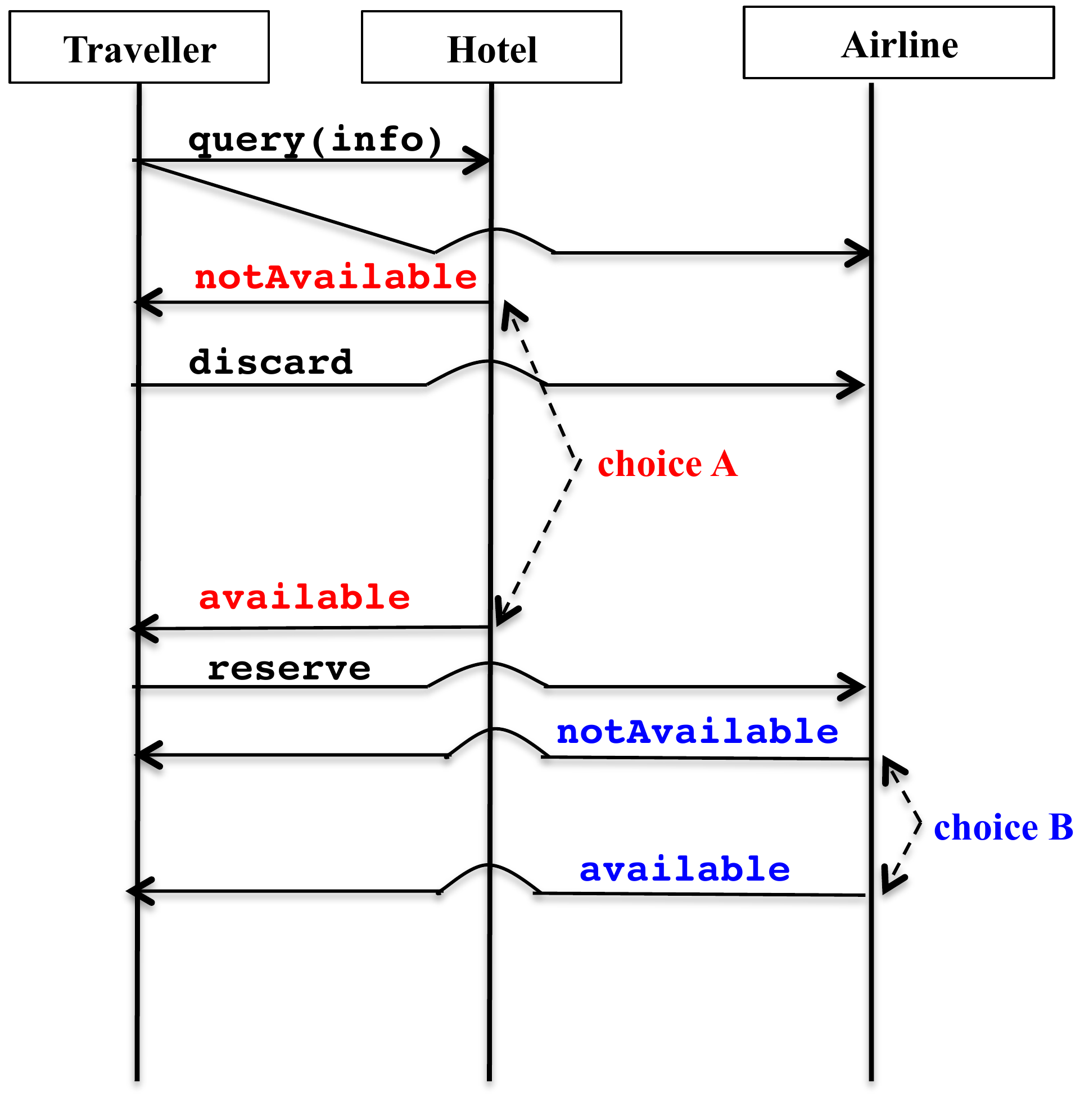}
\hspace*{\fill}
\caption{ Traveller planning a trip.}\label{f:exampleUML} \label{fig:exampleUML}
\end{wrapfigure}
\indent
Our focus is in the context of structured communications, more precisely
multiparty
sessions, see~\cite{HVK98,CHY08}. The choreography of communications is described by global types,
which are projected on the participants to get their interaction patterns~\cite{CDPY16,DGJPY16}.
In order to fix the points of the computations we may revert to, we add checkpoints to
the syntax of global types. In contrast to previous work, see~\cite{BDLdL15,TY15,TY16}, our checkpoints are named and
rollbacks specify the name of the checkpoints to which we revert.

We illustrate our approach by discussing an example. Consider the UML sequence diagram of Figure~\ref{fig:exampleUML}. In this example, there are three interacting participants, 
named Traveller (\Traveler), Hotel (\Hotel)
and Airline (\Airline), that establish a session. \Traveler, planning a trip, sends a message labelled 
\query,  to \Hotel\ and
to \Airline\ with the details of his journey (abstracted as a string \info). 
\Hotel\ answers to \Traveler\ with either the message \notAvailable\
or \available. If \Hotel\ answers \notAvailable, \Traveler\ sends to \Airline\ a message \discard\ saying
to ignore the previous \query. If \Hotel\ answers \available, then \Traveler\ send a message to \Airline\
asking to \reserve\ flights for the journey, to which \Airline\ answers to \Traveler\ with either 
the message \notAvailable\ or \available.

The choice made by \Hotel, named $A$, and the one made by \Airline, named $B$, are checkpointed choices.
This means that the computation could revert to one of these points of the interaction and the 
given participant could make a different choice. Rolling back to a choice point involves all the participants which crossed this choice point. Moreover, rollback may
happen only when all the participants that have this choice point in their future
have crossed it. Typing ensures that not involved participants are terminated.

Assume that \Hotel, after sending the message \available\ to
\Traveler, discovers that instead for the required dates it is fully booked, and wants to roll back to the choice
point $A$. 
Before rolling back, it has to make sure that \Traveler\ has sent the \reserve\ message to
\Airline\ and \Airline\ has received the message. So that, they can all go back to the interaction
before $A$. Otherwise, if \Traveler\ has not sent the \reserve\ message, and \Hotel\ goes back to
sending the message \notAvailable\ to \Traveler\ which is not expecting a message from \Hotel, there could be an
unpleasant misunderstanding, which is formally represented by a ``stuck'' computation.

{\bf Outline}
Section~\ref{c} introduces our calculus, completed by the type system of Section~\ref{ts}. Subject reduction, session fidelity and progress are proved in Section~\ref{mp}.  Section~\ref{rwc} discusses related papers and future work. 


\mysection{Calculus}\label{c}
In this section we introduce the syntax and the semantics of multiparty sessions with {\em named checkpoints}. 
\mysubsection{Syntax}\labelx{syntaxSection}
A \emph{multiparty session} is a series of 
interactions between a fixed number of
participants, possibly with branching and recursion~\cite{CHY08}.

We use the following base sets:  \emph{values}, ranged over by $\val,\val',\ldots$;
 \emph{expressions}, ranged over by $\e,\e',\ldots$;
\emph{expression variables}, ranged over by
$x,y,z\dots$; \emph{labels}, ranged over by $\ell,\ell',\dots$; {\em checkpoint names}, ranged over by $A,B,\ldots$; 
 \emph{session participants}, ranged over by $\pp,\q,\ldots$;
\emph{process variables}, ranged over by $X,Y,\dots$; 
\emph{processes}, ranged over by $P,Q,\dots$; {\em configurations}, ranged over by $\PrConf, \PrConf', \ldots$; \emph{multiparty sessions}, ranged over by $\Nw,\Nw',\dots$; \emph{networks}, ranged over by $\Nt,\Nt',\dots$.

Our processes are obtained from the processes of~\cite{DGJPY16}  by adding named checkpoints before external and internal choices. 

\begin{definition}\labelx{procdef}
 \emph{Processes} are defined by:
{\small 
\myformula{$
\begin{array}{ll}
\P  ::=  & \extChoice{i}{I}{\pp}{\ell}{x}{\P}~~\sep~~\intChoice{i}{I}{\pp}{\ell}{\e}{\P} 
~~\sep~~ \CkOne{\extChoice{j}{J}{\pp}{\ell}{x}{\P}}{A}
~~\sep~~ \CkOne{\intChoice{j}{J}{\pp}{\ell}{\e}{\P}}{A}
~~\sep~~   \mu\recvar.\P ~~\sep~~  \recvar ~~\sep~~\inact
\end{array}
$}
}
\end{definition} 
We say that $\CkOne\P A$ is a {\em process checkpointed by $A$}.

The input process $\extChoice{i}{I}{\pp}{\ell}{x}{\P}$ waits for a value and  a label $\ell_i$ with $i\in I$ from participant $\pp$ and
the output process $\intChoice{i}{I}{\q}{\ell}{\e}{\P}$ sends the value of an expression $\e_i$ and a label $\ell_i$ with $i\in I$  to participant $\q$. Checkpointed input and output processes behave in a similar way, except that, when sending/reading a message the checkpointed process is memorised, and can be
executed again after a rollback. As usual, in writing processes
we omit trailing $\inact$'s, and empty parameters.

\begin{example}\label{ex:syntax}
Consider the example of Figure~\ref{fig:exampleUML}. The processes associated with the participants \Traveler,
\Hotel, and \Airline\ are defined as follows (we abbreviate the labels of messages with their first two consonants):

\hspace{-25pt}
{\small 
\begin{tabular}{lll}
$\P_{\Traveler}$&=&$\sendMsg{\Hotel}{\queryM}{\infoM}{\sendMsg{\Airline}{\queryM}{\infoM}
{\P'_{\Traveler}}}\ $
where $\ \P'_{\Traveler}=\CkOne{\extChPr{
                   \rcvMsgS{\Hotel}{\notAvailableM}{\sendPrcS{\Airline}{\discardM} }
                    \ ,\ 
 \rcvMsgS{\Hotel}{\availableM}{\sendMsgS{\Airline}{\reserveM}{
                    \CkOne{ 
                    \extChPr{
                    \rcvPrcS{\Airline}{\notAvailableM}
                    \ ,\ 
                    \rcvPrcS{\Airline}{\availableM}}
                     }
                    {B}
                    }}
        }
   }{A}$\\
$\P_{\Hotel}$&=&$\rcvMsg{\Traveler}{\queryM}{x}{\CkTwo{\intChPr{
\sendPrcS{\Traveler}{\notAvailableM}
\ ,\  \sendPrcS{\Traveler}{\availableM}
}}{A}}$\\
$\P_{\Airline}$&=&$\rcvMsg{\Traveler}{\queryM}{x}{\P'_{\Airline}}\ $ where $\ \P'_{\Airline}=
          \CkOne{\extChPr{
            \rcvPrcS{\Traveler}{\discardM}\ ,\
            \rcvMsgS{\Traveler}{\reserveM}{
              \CkTwo{
              \intChPr{
 \sendPrcS{\Traveler}{\notAvailableM}                       \ ,\  
              \sendPrcS{\Traveler}{\availableM}
              }
              }{B}
             }
           }}{A}$
           \end{tabular}
}
\end{example}

In order to allow backward reductions,
the configurations of session participants contain both  active processes and sequences of checkpointed internal and external choices, denoting the processes that should run in case of rollbacks. 
\begin{definition}
{\em Configurations}, ranged over by $\PrConf$, are pairs $\pairPr{\BP}{\P}$, where $\P$ is a process, {\em the active process},
and 
{\myformula{$
\BP::=\noBP~  \sep  ~\Push{\BP}{\CkOne{\P}{A}}
$}
}
is a (possibly empty) sequence of checkpointed processes, dubbed {\em checkpointed sequence}. 
\end{definition}
In the sequence $\Push{\BP}\P$ we call $\P$ the {\em top process}.

 {\em Multiparty sessions}, ranged over by $\Nw$, are parallel compositions of pairs participant/configuration (denoted by $\confEl{\pp}{\PrConf}$):  \myformula{$\Nw::=\confEl{\pp}{\PrConf}~  \sep  ~\Nw\pc\Nw$}

{\em Networks}, ranged over by $\Nt$, are parallel composition of sessions: \myformula{$\Nt::=\Nw~  \sep  ~\Nt\PC\Nt$}

\mysubsection{Operational Semantics}\labelx{opSemanticsSection}

The LTS of configurations is given in Figure~\ref{rrc}. The forward rules are as expected, only internal choices with more than one branch can silently reduce. 
When the active process crosses a checkpoint, it is memorised at the top of the checkpointed sequence (rules \rln{[\ck\chc]} and \rln{[\ck Rcv]}). The backward rule can choose as the new active process an arbitrary process $\P$ in the current checkpointed sequence: the name of the checkpoint of $\P$ decorates the transition (rule \rln{[\rb P]}). This is essential in order to guarantee that the backward reduction of multiparty sessions produces well-behaved sessions, see rule \rln{[\rb M]} in Figure~\ref{rrsn}. 

\begin{figure}[h] 
{\small 
\begin{center}
$
\begin{array}{c}
\myruleOneLine{\pairPr{\BP}{\intChoice{i}{I}{\pp}{\ell}{\e}{\P}}\stackred{\silent}\pairPr{\BP}{\sendMsg{\pp}{\ell_k}{\e_k}{\P_k}}\quad k\in I\not=\set k}{[\chc]} \\ \\
\myruleOneLine{\pairPr{\BP}{\CkOne{\intChoice{j}{J}{\pp}{\ell}{\e}{\P}}{A}}\stackred{\silent}
\pairPr{\Push{\BP}{\CkOne{\intChoice{j}{J}{\pp}{\ell}{x}{\P}}{A}}}{\sendMsg{\pp}{\ell_k}{\e_k}{\P_k}}\quad k\in J\not=\set k}
{{[\ck\chc]}}\\ \\
\myruleOneLine{\pairPr{\BP}{\sendMsg{\pp}{\ell}{\e}{\P}}\stackred{\pp!\ell(\val)}\pairPr{\BP}{\P}\quad \eval{\e}{\val}}{[{\snd}]} \\ \\
\myruleOneLine{\pairPr{\BP}{\extChoice{i}{I}{\pp}{\ell}{x}{\P}}\stackred{\pp?\ell_j(\val)}\pairPr{\BP}{\subst{\P_j}{x}{\val}}\quad j\in I}{[{Rcv}]}\\ \\
\myruleOneLine{\pairPr{\BP}{\CkOne{\extChoice{j}{J}{\pp}{\ell}{x}{\P}}{A}}\stackred{\pp?\ell_k(\val)}\pairPr{\Push{\BP}{\CkOne{\extChoice{j}{J}{\pp}{\ell}{x}{\P}}{A}}}{{\subst{\P_k}{x}{\val}}}\quad k\in J}{[CkRcv]}\\ \\
\myruleOneLine{\pairPr{\Push{\Push{\BP}{\CkOne{{\P}}{A}}}{\BP'}}{\P'}\stackred{A}\pairPr{\BP}{\CkOne{{\P}}{A}}}{[\rb P]}
\end{array}
$
\end{center}
}
\caption{ Reduction rules of configurations.}\label{rrc}
\end{figure}

The operational semantics of sessions and network is shown in Figure~\ref{rrsn}, where $\alpha$ ranges over $\silent, \pp!\ell(\val), \pp?\ell(\val), A$. This semantics relies on a structural equivalence $\equiv$ for which the parallel operators $\pc$ and $\PC$ are commutative and associative and have $\confEl{\pp}{\pairPr\epsilon\inact}$ as neutral element. 

The only interesting rule is rule  \rln{[\rb M]}. In this rule we use the mapping $\A$, that associates to a configuration  the set of 
the checkpoint names
of processes belonging to its checkpointed sequence. Formally:  
   \myformula{$
 \A(\pairPr{\BP}{\inact})=\A(\BP)\qquad\A(\noBP)=\emptyset\qquad \A(\Push{\BP}{\CkOne{\P}{A}})={\A(\BP)}\cup\set A
$}
This mapping is defined only for configurations having $\inact$ as their active process. This is enough since the typing rules ensure that the processes which did not traverse some checkpoints are terminated. 
A multiparty session can roll back to processes at the checkpoint named $A$ only if all the sets $\A$ of the configurations which remain unchanged are defined (i.e. $\inact$ is the active process of these configurations) and they do not contain $A$. 

\begin{figure}[h] 
{\small \begin{center}
$
\begin{array}{c}
\myrule{\PrConf\stackred{\alpha}\PrConf'}
{\confEl{\pp}{\PrConf}\stackred{\alpha}\confEl{\pp}{\PrConf'} }{[PC]}\qquad 
\myrule{\Nw\stackred{~~\silent~~}\Nw'}
{\Nw\pc\Nw''\stackred{~~\silent~~}\Nw'\pc\Nw'' }{[PrM]}\qquad 
  \myrule{\Nw_1\equiv\Nw'_1\quad\Nw'_1\stackred{~~\silent~~}\Nw'_2\quad\Nw'_2\equiv\Nw_2}
{\Nw_1\stackred{~~\silent~~}\Nw_2 }{[EqM]}\\ \\
\myrule{\confEl{\pp}{\PrConf_\pp}\stackred{\q!\ell(\val)}\confEl{\pp}{\PrConf'_\pp}
\quad\quad
\confEl{\q}{\PrConf_\q}\stackred{\pp?\ell(\val)}\confEl{\q}{\PrConf_\q'}
}
{\confEl{\pp}{\PrConf_\pp}\pc\confEl{\q}{\PrConf_\q}
\stackred{~~\silent~~}
\confEl{\pp}{\PrConf_\pp'}\pc\confEl{\q}{\PrConf_\q'}
}{[Com]}\\ \\
\myrule{\confEl{\pp_i}{\PrConf_{\pp_i}}\stackred{A}\confEl{\pp_i}{\PrConf'_{\pp_i}}~~\forall i\in I
\quad\quad
A\not\in\A(\PrConf_{\q_j})~~\forall j\in J
}
{\Pi_{i\in I}{\confEl{\pp_i}{\PrConf_{\pp_i}}\pc\Pi_{j\in J}\confEl{\q_j}{\PrConf_{\q_j}}
\stackred{~~\silent~~}
\Pi_{i\in I}\confEl{\pp_i}{\PrConf_{\pp_i}'}\pc\Pi_{j\in J}\confEl{\q_j}{\PrConf_{\q_j}}
}}{[\rb M]}\\\\
\myrule{\Nt\stackred{~~\silent~~}\Nt'}
{\Nt\PC\Nt''\stackred{~~\silent~~}\Nt'\PC\Nt'' }{[PrN]}\qquad \qquad \qquad \qquad
 \myrule{\Nt_1\equiv\Nt'_1\quad\Nt'_1\stackred{~~\silent~~}\Nt'_2\quad\Nt'_2\equiv\Nt_2}
{\Nt_1\stackred{~~\silent~~}\Nt_2 }{[EqN]}\\\ \\
\end{array}
$
\end{center}
}
\caption{Reduction rules of sessions and networks.}\labelx{rrsn}
\end{figure}

In networks the different sessions reduce independently. For this reason the same participant can interact in different sessions belonging  to the same network. We use $\stackredstar{~~\silent~~}$ to denote the transitive
and reflexive closure of the $\stackred{~~\silent~~}$ relation.

\begin{example}\label{ex:opSemantics}
Consider the processes of Example~\ref{ex:syntax}. We first give some reductions possible
for the configurations of the three participants starting from empty checkpointed sequences.
{\small 
\myformula{$
\begin{array}{lllr}
{\pairPr{\noBP}{\P_{\Traveler}}}& \stackred{\Hotel!\queryM(\infoM)}
     {\pairPr{\noBP}
             {\sendMsg{\Airline}{\queryM}{\infoM}
{\P'_{\Traveler}}}=\PrConf_1
     }   & \rln{[\snd]} & (1)\\
   & \stackred{\Airline!\queryM(\infoM)}
   {\pairPr{\noBP}
             {\P'_{\Traveler}}=\PrConf_2
     }   & \rln{[\snd]}& (2)\\ 
& \stackred{\Hotel?\availableM}
   {\pairPr{\P'_{\Traveler}}
             {{\sendMsgS{\Airline}{\reserveM}{
                    (\CkOne{ 
                    \extChPr{\rcvPrcS{\Airline}{\notAvailableM}\ ,\ \rcvPrcS{\Airline}{\availableM}}
                     }
                    {B})
                    }}}=\PrConf_3
     }   & \rln{[CkRcv]}& (3)\\ 
& \stackred{\Airline!\reserveM}
   {\pairPr{\P'_{\Traveler}}
             {\CkOne{ 
                    \extChPr{\rcvPrcS{\Airline}{\notAvailableM}\ ,\ \rcvPrcS{\Airline}{\availableM}}
                     }
                    {B}
                    }=\PrConf_4
     }   & \rln{[\snd]}& (4)\\ 
& \stackred{\Airline?\notAvailableM}
   {\pairPr{\Push{\P'_{\Traveler}}{(\CkOne{ 
                    \extChPr{\rcvPrcS{\Airline}{\notAvailableM}\ ,\ \rcvPrcS{\Airline}{\availableM}}
                     }
                    {B})}}
             {\inact}=\PrConf_5
     }   & \rln{[CkRcv]}& (5)
\end{array}
$}
}
{\small 
\myformula{$
\begin{array}{lllr}    
{\pairPr{\noBP}{\P_{\Hotel}}}& \stackred{\Traveler?\queryM(\infoM)}
     {\pairPr{\noBP}
             {\CkTwo{\intChPr{\sendPrcS{\Traveler}{\notAvailableM},\sendPrcS{\Traveler}{\availableM}  }}{A}}=\PrConf_6
     }   & \rln{[Rcv]} & (6)\\
   & \stackred{\silent}
   {\pairPr{\CkTwo{\intChPr{\sendPrcS{\Traveler}{\notAvailableM},\sendPrcS{\Traveler}{\availableM}  }}{A}}
             {\sendPrcS{\Traveler}{\availableM}}=\PrConf_7
     }   & \rln{[CkChc]}& (7)\\  
   & \stackred{\Traveler!\availableM}
{\pairPr{\CkTwo{\intChPr{\sendPrcS{\Traveler}{\notAvailableM},\sendPrcS{\Traveler}{\availableM}  }}{A}}
             {\inact}=\PrConf_8
     }   & \rln{[\snd]}& (8)
\end{array}
$}
}
{\small 
\myformula{$
\begin{array}{lllr}    
{\pairPr{\noBP}{\P_{\Airline}}}& \stackred{\Traveler?\queryM(\infoM)}
     {\pairPr{\noBP}
             {\P'_{\Airline}}=\PrConf_9
     }   & \rln{[Rcv]} & (9)\\
   & \stackred{\Traveler?\reserveM}
   {\pairPr{\P'_{\Airline}}
             {
              \CkTwo{
              \intChPr{\sendPrcS{\Traveler}{\notAvailableM},\sendPrcS{\Traveler}{\availableM}  }
              }{B}}=\PrConf_{10}
     }   & \rln{[CkRcv]}& (10)\\  
   & \stackred{\silent}
   {\pairPr{\Push{\P'_{\Airline}}{\CkTwo{\intChPr{\sendPrcS{\Traveler}{\notAvailableM},\sendPrcS{\Traveler}{\availableM}  }}{B}}}
             {\sendPrcS{\Traveler}{\availableM}}=\PrConf_{11}
     }   & \rln{[CkChc]}& (11)\\  
   & \stackred{\Traveler!\availableM}
{\pairPr{\Push{\P'_{\Airline}}{\CkTwo{\intChPr{\sendPrcS{\Traveler}{\notAvailableM},\sendPrcS{\Traveler}{\availableM}  }}{B}}}
             {\inact}=\PrConf_{12}
     }   & \rln{[\snd]}& (12)           
\end{array}
$}
}
Starting from the initial session $\Nw= \confEl{\Traveler}{\pairPr{\noBP}{\P_{\Traveler}}}~  \sep  ~ \confEl{\Hotel}{\pairPr{\noBP}{\P_{\Hotel}}} ~  \sep  ~ \confEl{\Airline}{\pairPr{\noBP}{\P_{\Airline}}}$, in which all participants have their associated processes as active processes and empty checkpointed sequences, we can have the 
reductions shown in Figure~\ref{r}.

\begin{figure}[h]
{\small 
{$
\begin{array}{lll}
\Nw  & \stackred{\silent}
\confEl{\Traveler}{\PrConf_1}~  \sep  ~ \confEl{\Hotel}{\PrConf_6} ~  \sep  ~ \confEl{\Airline}{\pairPr{\noBP}{\P_{\Airline}}}& \mbox{\rm using (1) and (6) and rules \rln{[Pc]}, \rln{[Com]}, and \rln{[PrM]}}\\
& \stackred{\silent} \confEl{\Traveler}{\PrConf_2}~  \sep  ~ \confEl{\Hotel}{\PrConf_6} ~  \sep  ~ \confEl{\Airline}{\PrConf_9} & \mbox{\rm using (2) and (9) and rules \rln{[Pc]}, \rln{[Com]}, and \rln{[PrM]}}\\
& \stackred{\silent} \confEl{\Traveler}{\PrConf_2}~  \sep  ~ \confEl{\Hotel}{\PrConf_7} ~  \sep  ~ \confEl{\Airline}{\PrConf_9} & \mbox{\rm using (7) and rules \rln{[Pc]}, and \rln{[PrM]}}\\
& \stackred{\silent} \confEl{\Traveler}{\PrConf_3}~  \sep  ~ \confEl{\Hotel}{\PrConf_8} ~  \sep  ~ \confEl{\Airline}{\PrConf_9} & \mbox{\rm using (3) and (8) and rules \rln{[Pc]}, \rln{[Com]}, and \rln{[PrM]}}\\
& \stackred{\silent} \confEl{\Traveler}{\PrConf_4}~  \sep  ~ \confEl{\Hotel}{\PrConf_8} ~  \sep  ~ \confEl{\Airline}{\PrConf_{10}} & \mbox{\rm using (4) and (10) and rules \rln{[Pc]}, \rln{[Com]}, and \rln{[PrM]}}\\
& \stackred{\silent} \confEl{\Traveler}{\PrConf_4}~  \sep  ~ \confEl{\Hotel}{\PrConf_8} ~  \sep  ~ \confEl{\Airline}{\PrConf_{11}} & \mbox{\rm using (11) and rules \rln{[Pc]}, and \rln{[PrM]}}\\
& \stackred{\silent} \confEl{\Traveler}{\PrConf_5}~  \sep  ~ \confEl{\Hotel}{\PrConf_8} ~  \sep  ~ \confEl{\Airline}{\PrConf_{12}} & \mbox{\rm using (5) and (12) and rules \rln{[Pc]}, \rln{[Com]}, and \rln{[PrM]}}\\
\end{array}
$}
}
\caption{Example of multiparty session reductions.}\labelx{r}
\end{figure}
From the final session  
we can do a rollback to $B$ as follows:
{\small 
\myformula{$
\prooftree
{
{\prooftree
\PrConf_5\stackred{B}\PrConf_4\quad\rln{[RbP]}
\justifies
\confEl{\Traveler}{\PrConf_5}\stackred{B}\confEl{\Traveler}{\PrConf_4}
\using
\rln{[Part]}
\endprooftree
}
\quad
{\prooftree
\PrConf_{12}\stackred{B}\PrConf_{10}\quad\rln{[RbP]}
\justifies
\confEl{\Airline}{\PrConf_{12}}\stackred{B}\confEl{\Airline}{\PrConf_{10}}
\using
\rln{[Part]}
\endprooftree
}
\quad
B\not\in\A(\PrConf_{8})
}
\justifies
\confEl{\Traveler}{\PrConf_5}~  \sep  ~ \confEl{\Hotel}{\PrConf_8} ~  \sep  ~ \confEl{\Airline}{\PrConf_{12}}
\stackred{\silent} 
\confEl{\Traveler}{\PrConf_4}~  \sep  ~ \confEl{\Hotel}{\PrConf_8} ~  \sep  ~ \confEl{\Airline}{\PrConf_{{10}}}
\using
\rln{[RbM]}
\endprooftree
$}
}
In a similar way, we can do a rollback to $A$ from the session 
$\confEl{\Traveler}{\PrConf_5}~  \sep  ~ \confEl{\Hotel}{\PrConf_8} ~  \sep  ~ \confEl{\Airline}{\PrConf_{12}}$ producing $\confEl{\Traveler}{\PrConf_2}~  \sep  ~ \confEl{\Hotel}{\PrConf_{{6}}} ~  \sep  ~ \confEl{\Airline}{\PrConf_{9}}$.
\end{example}


\mysection{Type System}\label{ts}
\mysubsectionA{Types}
{\em Sorts}  are ranged over by $\ST$ and defined by:\qquad
$ \ST    \quad   ::=   \quad                         \tint \sep\tbool\sep\ldots$

Single-threaded global types describe the whole conversation scenarios of multiparty sessions, when they reduce forward. The communications can be either without or with named checkpoints. 

Global types instead take into account both forward and backward reductions of multiparty sessions. They have therefore a structure which mimics the structure of configurations. 

\begin{definition}\labelx{sgtypesdef}
\begin{myenumerate}{}
\item\labelx{sgtypesdef1} {\em Single-threaded global types
} are defined by:
{\small 
\myformula{$\G~~ ::= ~~\Gvti\pp \ell \ST {\G} ~\sep~\CkOne{\Gvtj{\pp}{\ell}{\ST}{\G}}{A} ~\sep ~\mu\ty.\G ~\sep ~\ty ~\sep ~\End
$}
} 
\item\labelx{sgtypesdef2} {\em Global types} are pairs $\pairGT\GG\G$, where $\GG$ is a (possibly empty) sequence of single-threaded global types with checkpoints having distinct names:
\myformula{$\GG::=\s~\sep ~\GG\cdot\CkOne\G A$}
\end{myenumerate}
\end{definition}

We say that $\G$ is the {\em active type} of $\pairGT\GG\G$. The condition in point~(\ref{sgtypesdef2}) of previous definition ensures that, if $\pairGT\GG\G$ is a global type, and $\GG=\Push{\Push{\GG'}{\CkOne{\G'} A}}{\GG''}$, then no single-threaded global type in $\GG'$, $\GG''$ can be checkpointed by $A$.

{\em Session types} correspond to projections of single-threaded global types onto the individual participants. Therefore, they can be decorated by named checkpoints. 
Inspired by~\cite{P11}, we use intersection and union types instead of standard branching and selection, see~\cite{CHY08}, to take advantage of the subtyping induced by subset inclusion.
The grammar of session types, ranged over by $\T$, is then 
{\small 
\myformula{$
\begin{array}{ll}
\T ~~  ::= & \intMsg{i}{I}{\pp}{\ell}{S}{\T} ~  \sep  ~ 
\unMsg{i}{I}{\pp}{\ell}{S}{\T} 
~  \sep  ~ \CkOne{\intMsg{j}{J}{\pp}{\ell}{S}{\T}}{A} ~\sep
\CkTwo{\unMsg{j}{J}{\pp}{\ell}{S}{\T}}{A} 
~  \sep  ~  \mu\ty.\T  ~\sep  ~ \ty   ~\sep ~  \End 
\end{array}
$}
}
In both global and session types we require  that:
\begin{myitemize}
\item $I,J$ are not empty sets and $J$ is not a singleton; 
\item $\ell_h\not=\ell_k$ if $h,k\in I$ or $h,k\in J$;
\item the name $A$ does not occur in a type checkpointed by $A$;
\item recursion is guarded.
\end{myitemize}
We constrain $J$ to contain at least two elements since it makes sense to reverse only when an alternative branch could be taken. 

Recursive types with the same regular tree are considered equal~\cite[Chapter 20, Section 2]{pier02}. In writing types we omit unnecessary brackets, intersections, unions, and $\End$. 

We extend the original definition of projection of single-threaded global types onto participants of~\cite{CHY08} in the line of~\cite{DGJPY16}. This generalisation allows session participants to  behave differently in  alternative branches of the same global type, after they have received a message identifying the branch. 
We define the partial operator $\interRcv$ from sets of session types - all uncheckpointed or checkpointed by the same name - to (possibly checkpointed) intersection types. 
The operator is defined when the set of session types contains only intersection of types with same sender and different labels. Otherwise, it is undefined. More precisely:
 \myformula{$\interRcv(\set{\T_i}_{i\in I})=\begin{cases}
 \bigwedge_{i\in I}\T_i     & \text{if there is }\pp\text{ such that }  \T_i=\intMsgLong{h}{H_i}{\pp}{\ell^{(i)}}{S^{(i)}}{\T_h^{(i)}}\text{ for all }i\in I,\\
 &\text{and }\ell_h^{(i)}\not=\ell_k^{(j)}
 \text{for all }h\in H_i, k\in H_j \text { and } i,j\in I, i\not=j\\
 \CkOne{\interRcv(\set{\T'_i}_{i\in I})}A     & \text{if }\T_i =\CkOne{\T'_i}A \text{ for all }i\in I\\
 \text{undefined}     & \text{otherwise}.
\end{cases}
$}
Notice that in defining $\interRcv$ we could allow identical types in $\T_i$ and $\T_j$ by the idempotence of intersection. We prefer the current choice for its simplicity. 

Figure~\ref{projection} gives the 
projection of single-threaded global types 
onto participants. Notice that, the projection of a checkpointed type onto a participant not involved in the initial communication is defined  only when it receives uncheckpointed messages in all the branches. For this reason, and since $\interRcv$ is a partial operator, projections onto some participants may be undefined. A single-threaded  global type $\G$ is {\em well formed} if all the projections of $\G$ onto all participants are defined. In the following we assume that all single-threaded global types are well formed.

\begin{figure}[h] 
{\small \begin{center}
$\proj{({\Gvti{\pp}{\ell}{\ST}{\G}})}\pr=\begin{cases}
   {\unMsgLong{i}{I}{\q}{\ell}{S}{\proj{\G_i}{\pr}}} & \text{if }\pr=\pp \\
   {\intMsgLong{i}{I}{\pp}{\ell}{S}{\proj{\G_i}{\pr}}}  & \text{if }\pr=\q\\
   \End & \text{if $\pr \neq \pp$ and $\pr \neq\q$ and $\proj{\G_i}{\pr}=\End$ for all $i\in I$}\\
\interRcv(\{\proj{\G_{i}}{\pr}\mid i\in I\})  & \text{if $\pr \neq \pp$ and $\pr \neq\q$}
\end{cases}$
\\[2mm]
$\proj{(\CkGl{\Gvtj{\pp}{\ell}{\ST}{\G}}{A})}\pr=\begin{cases}
   \CkTwo{\unMsgLong{j}{J}{\q}{\ell}{S}{\proj{\G_j}{\pr}}}{A} & \text{if }\pr=\pp \\
   \CkOne{\intMsgLong{j}{J}{\pp}{\ell}{S}{\proj{\G_j}{\pr}}}{A}  & \text{if }\pr=\q\\
    \End  & \text{if $\pr \neq \pp$ and $\pr \neq\q$ and $\proj{\G_{j}}{\pr}=\End$ for all $j\in J$}\\
      \CkOne{\interRcv(\{\proj{\G_{j}}{\pr}\mid j\in J\})}{A}  & \text{if $\pr \neq \pp$ and $\pr \neq\q$}\\
  & \text{and $\interRcv(\{\proj{\G_{j}}{\pr}\mid j\in J\})$ is not checkpointed}\\
\end{cases}$
\\[2mm]
 $\proj{(\mu\ty.\G)}\pp=\begin{cases}
  \mu\ty. \proj{\G}{\pp}   & \text{if }\pp\in\G, \\
   \End   & \text{otherwise}.
\end{cases}$~~~~~$\proj{\ty}{\pp} = \ty $~~~~~~~$\proj{\End}{\pp} = \End $ 
\end{center}
}
\caption{Projection of single-threaded  global types onto participants.}\labelx{projection}
\end{figure}

\begin{example}\label{ex:globalType}
Assuming that $\stringT$ is the sort of strings, the global type for the interaction of Figure~\ref{fig:exampleUML} is $\G$ defined as follows:

\begin{tabular}{c}
$\G=\GlComm{\Traveler}{\Hotel}\queryM(\stringT).\GlComm{\Traveler}{\Airline}\queryM(\stringT).\G_1$ where\\
$\G_1=\CkGlComm{\Hotel}{\Traveler}{A}\{\notAvailableM.\GlComm{\Traveler}{\Airline}\discardM\ ,\ 
           \availableM.\GlComm{\Traveler}{\Airline}\reserveM.\G_2\}$ and
$\G_2=\CkGlComm{\Airline}{\Traveler}{B}\{\notAvailableM\ ,\ \availableM\}$  
\end{tabular}

\noindent
$\G$ is well formed since the projections onto its participants are:

\begin{tabular}{lll}
$\proj{\G}{\Traveler}$&=&$\sendMsg{\Hotel}{\queryM}{\stringT}
           {\sendMsg{\Airline}{\queryM}{\stringT} 
              {\CkOne{\interTy{
              \rcvMsgS{\Hotel}{\notAvailableM}{} {{\sendPrcS{\Airline}{\discardM}}}\ ,\
                  \rcvMsgS{\Hotel}{\availableM}{\sendMsgS{\Airline}{\reserveM}
                  {\CkOne{\interTy{\rcvPrcS{\Airline}{\notAvailableM}\ ,\ \rcvPrcS{\Airline}{\availableM}}}{B}}
                  }}}{A}}}
                  $\\
$\proj{\G}{\Hotel}$&=&$\rcvMsg{\Traveler}{\queryM}{\stringT}{
                 \CkTwo{\unionTy{
                 \sendPrcS{\Traveler}{\notAvailableM}\ ,\ \sendPrcS{\Traveler}{\availableM}
                 }}{A}
                 }$\\
$\proj{\G}{\Airline}$&=&$\rcvMsg{\Traveler}{\queryM}{\stringT}
           {\CkOne{\interTy{\rcvPrcS{\Traveler}{\discardM}\ ,\
                 \rcvMsgS{\Traveler}{\reserveM}
                 {\CkTwo{\unionTy{
                 \sendPrcS{\Traveler}{\notAvailableM}\ ,\ \sendPrcS{\Traveler}{\availableM}
                 }}{B}}
                  }}{A}}$
                  \end{tabular}
\end{example}

In order to type checkpointed sequences and configurations we need (possibly empty) {\em sequences of checkpointed session types}, ranged over by $\RT$:
\myformula{$
\RT::=\s~  \sep  ~\Push{\RT}{\CkOne{\T}{A}}
$}
and pairs $\pairT{\RT}{\T}$, dubbed {\em configuration types}.

The typing is made more flexible by a subtyping relation on session types exploiting the standard inclusions for intersection and union. A checkpointed type can be a subtype only of a type checkpointed by the same name.   Figure~\ref{fig:subt} gives the subtyping rules: the double line in rules indicates
that the rules are interpreted {\em coinductively}~\cite[Chapter 21]{pier02}. Subtyping can be easily decided, see for example~\cite{GH05}. 

\begin{figure}[h]
{\small \centerline{$
\begin{array}{ccc}
\inferrule[\rulename{sub-end}]{}
  {\End \subt \End}&\qquad&
  \cinferrule[\rulename{sub-ck}]{\T\subt\T'}
  {\CkTwo{\T}{A}\subt\CkTwo{\T'}{A}}
  \\ \\
\cinferrule[\rulename{sub-in}]{
\forall i\in I: \quad \T_i \subt\T_i' }{
  \intMsg{i}{I\cup J}{\pp}{\ell}{S}{\T} \subt \intMsg{i}{I}{\pp}{\ell}{S}{\T'}
}
&\qquad&
\cinferrule[\rulename{sub-out}]{
 \forall i\in I: \quad \T_i \subt \T'_i  
}{\unMsg{i}{I}{\pp}{\ell}{S}{\T}
  \subt\unMsg{i}{I\cup J}{\pp}{\ell}{S}{\T'}
 }
\end{array}
$}
}
\caption{\label{fig:subt} Subtyping rules.}
\end{figure}

\mysubsection{Typing Rules} 
We distinguish six  kinds of typing judgements
\myformula{$
  \dere\Gamma\e\ST \quad\quad \derP\Gamma\P\T \quad\quad \derN\BP\RT\quad\quad \derN\PrConf{\pairT{\RT}{\T}}\quad\quad \derN\Nw\pairGT{\GG}{\G}\quad\quad \derNok\Nt
$}
where $\Gamma$ is the environment $\Gamma ::= \emptyset \sep \Gamma, x:\ST \sep \Gamma, X:\T$ that associates expression variables with sorts and process variables with session types. 

Figure~\ref{fig:typingPr} gives the typing rules for processes. Processes typing exploits the correspondence between external choices and intersections, internal choices and unions. A checkpointed process has a type checkpointed by the same name. 

\begin{figure}[h]
{\small 
\centerline{$
\begin{array}{ccc}
 \myrule{\derP {\Gamma,x:\ST}{\P_i}{\T_i}
  }{\derP \Gamma{\extChoice{i}{I}{\pp}{\ell}{\e}{\P}}{\intMsg{i}{I}{\pp}{\ell}{S}{\T}}}{[t-In]}
  &\qquad&
  \myrule{\dere\Gamma{\e_i}{\ST_i}~~\  \derP {\Gamma}{\P_i}{\T_i}
   }{\derP \Gamma{\intChoice{i}{I}{\pp}{\ell}{\e}{\P}}{\unMsg{i}{I}{\pp}{\ell}{S}{\T}}}{[t-Out]}
  \\\\
  \myrule{\derP {\Gamma}{\P}{\T}}
  {\derP \Gamma{\CkOne{\P}A}{\CkOne\T A}}{[t-\ck]}
  &\qquad&
   \derP \Gamma{ \inact}\End~~\rln{[t-$\inact$]}
  \\\\
    \myrule{\derP {\Gamma,X:\T}{\P}{\T}}
  {\derP \Gamma{\mu X.\P}\T}{[t-Rec]}
&\qquad&
   \derP {\Gamma,X:\T}{X}{\T}~~\rln{[t-Var]}     \end{array}
$}
}
\caption{\label{fig:typingPr} Typing rules for processes.}
\end{figure}

\begin{figure}
{\small \centerline{$
\begin{array}{c}
\myrule{
\derS{\BP}{\RT}\quad\derP{}{\P}{\T}
}
{\derS{\Push\BP\P}{\PushT\RT\T}}
{[t-SP]}\qquad\myrule{
\derS{\BP}{\RT}\quad\derP{}{\P}{\T}
}
{\derN{\confEl{\pp}{\pairPr\BP\P}}{\pairT\RT\T}}
{[t-C]}\\ \\
\myrule{\derN{\confEl{\pp_i}{\PrConf_i}}{\pairT{\RT_i}{\T_i}} \quad\agr{\pairT{\RT_i}{\T_i}}{{\pp_i}}{\pairGT{\GG}{\G}}\quad i\in I\quad \len\GG=\text{max}\set{\len{\RT_i}\mid i\in I}
\quad \pt(\GG)\cup\pt(\G)\subseteq\set{\pp_i\mid i\in I}}
{\derN{\Pi_{i\in I}\confEl{\pp_i}{\PrConf_i}}{\pairGT{\GG}{\G}}}
{[t-M]}\\\\
\myrule{\derN{\Nw}{\pairGT{\GG}{\G}}}{\derNok{\Nw}}{[t-\checkmark]}\qquad
\myrule{\derNok{\Nt}\quad\derNok{\Nt'}}{\derNok{\Nt\PC\Nt'}}{[t-N]}
\end{array}
$}}
\caption{\label{fig:typingNet}Typing rules for checkpointed sequences, multiparty sessions and networks.}
\end{figure} 

Figure~\ref{fig:typingNet} gives the remaining typing rules. Sequences of checkpointed processes are typed by sequences of checkpointed types (rule \rln{[t-SP]}). Configurations are typed by configuration types (rule \rln{[t-C]}). 

The most interesting rule is rule \rln{[t-M]} for typing multiparty sessions.
The set $\pt(\G)$ of participants of a single-threaded global type is defined by
\begin{center}
$\pt(\Gvti\pp \ell \ST {\G})=\{\pp,\q\}\cup\pt(\G_i)~(i\in I)\footnotemark\quad
\pt(\CkOne{\G}{A})=\pt(\mu\ty.\G)=\pt(\G)\quad\pt(\End)=\pt(\ty)=\emptyset$
\end{center}
\footnotetext{The projectability of $\G$ ensures $\G_i=\G_j$ for all $i,j\in I$.}
\noindent
The definition is extended to sequences of checkpointed single-threaded global types by \myformula{$\pt(\s)=\emptyset\qquad\qquad\pt(\GG\cdot\CkOne\G A)=\pt(\GG)\cup\pt(\G)$}
The condition $\pt(\GG)\cup\pt(\G)\subseteq\set{\pp_i\mid i\in I}$ ensures the presence of all session participants and allows the
typing of sessions containing 
$\confEl{\pp}{\pairPr{\epsilon}{\inact}}$ (also if $\pp\not\in\pt(\GG)\cup\pt(\G)$), a property needed to guarantee invariance of types under structural equivalence.
Rule \rln{[t-M]} requires that the types of the active processes are subtypes of the projections of a unique global type. This condition must hold also after a rollback, in which all the processes checkpointed by $A$, in the respective checkpointed sequences, become the active processes (reduction rule \rln{[\rb M]}). For this reason we type a multiparty session by a global type such that the sequence of single-threaded global types has length equal to the maximum of the lengths of the sequences of session types in the types of configurations. This is expressed by the condition $\len\GG=\text{max}\set{\len{\RT_i}\mid i\in I}$ in the premise of rule \rln{[t-M]}, where $\len\GG$ is the length of the sequence $\GG$ and $\len\RT$ is the length of the sequence $\RT$. This requirement is clearly not enough. Further constraints are prescribed 
by the agreement between global types and configuration types of session participants. This agreement is made more flexible by the use of subtyping. 
\begin{definition}\label{dar}Let $\RT=\T_1\cdot\ldots\cdot\T_n$ and $\GG=\G_1\cdot\ldots\cdot\G_m$. 
The configuration type $\pairT{\RT}{\T}$ {\em $\pp$-agrees} with the global type $\pairGT\GG\G$ (notation $\agr{\pairT{\RT}{\T}}{\pp}{\pairGT\GG\G}$) if all the following conditions hold: 
\begin{myenumerate}{}
\item\label{dar1} $\T_i\leq\proj{\G_i}{\pp}$ for $1\leq i\leq n$;
\item\label{dar3} if $\T=\End$, then $n\leq m$ and $\proj{\G_i}{\pp}=\proj{\G}{\pp}=\End$ for $n+1\leq i\leq m$;
\item\label{dar4} if $\T$ is a union type, then $n=m$ and $\T\leq\proj{\G}{\pp}$;
\item\label{dar5} if $\T$ is an intersection type, then either $n=m$ and $\T\leq\proj{\G}{\pp}$ or $n=m-1$ and $\T\leq\proj{\G_m}{\pp}$ and $\T=\CkTwo{\T'}{A}$ 
 and $\T'\leq\proj{\G}{\pp}$. 
\end{myenumerate}
\end{definition}
Condition~\ref{dar1}, typing rule \rln{[t-\ck]}, and the fact that the checkpoints 
in $\GG$ have distinct names (see Definition~\ref{sgtypesdef}), ensure that, after a rollback, 
the processes becoming active are in the same position inside the checkpointed sequences.
 Moreover, these processes have types which are subtypes of the projections of the same single-threaded global type.
Condition~\ref{dar3} deals with the case in which the active process of $\pp$ is $\inact$. Once the active process of a participant is $\inact$, no more process are added to its checkpointed sequence, so, in case of a rollback to a checkpoint crossed afterwards, by some other participants, its active process would remain $\inact$. Therefore the projection of the global type should be $\End$.
Conditions~\ref{dar4} and~\ref{dar5} deal with the case in which the type of the active process of $\pp$
is a union or an intersection.  If $n=m$,  
then these conditions require that its type is subtype of the projection of the global type $\G$, i.e., $\T\leq\proj{\G}{\pp}$. If the active process of $\pp$ is typed by an intersection, then it must
be an external choice of input processes (rule \rln{[T-In]}). In this case the global type $\G$ could capture the situation in which a participant $\q$ has internally chosen one branch, 
memorised in the checkpointed sequence its active process, and the active process of $\q$ has became an
uncheckpointed output (reduction rule \rln{[\ck\chc]}). 
This means that, the length of the checkpointed sequence of $\q$ must be $m$, while the length of the 
checkpointed sequence of $\pp$ must be $m-1$.
To deal with this situation, since an uncheckpointed and a checkpointed type cannot be projections of the same global type, in condition~\ref{dar5}, we require that $\T\leq\proj{\G_m}{\pp}$ and $\T=\CkTwo{\T'}{A}$ and $\T'\leq\proj{\G}{\pp}$. 
Notice that, $\G$ must be uncheckpointed and $\G_m$ must be checkpointed by $A$. This condition is illustrated in Example~\ref{ex:typing}.

Rule \rln{[t-M]} requires that all types of the configurations which built the multiparty session agree with the global type of the session itself, conditions $\derN{\confEl{\pp_i}{\PrConf_i}}{\pairT{\RT_i}{\T_i}}$   and $\agr{\pairT{\RT_i}{\T_i}}{{\pp_i}}{\pairGT{\GG}{\G}}$ for  $i\in I$.

A network is well typed if and only if all its multiparty sessions are well typed. 

\begin{example}\label{ex:typing}
To show the typings for the networks produced by the reduction of Example~\ref{ex:opSemantics} consider the
global type $\G$ of Example~\ref{ex:globalType}. We can derive
\begin{myenumerate}{}
\item $\derN{\confEl{\Traveler}{\pairPr{\noBP}{\P_{\Traveler}}}~  \sep  ~ \confEl{\Hotel}{\pairPr{\noBP}{\P_{\Hotel}}} ~  \sep  ~ \confEl{\Airline}{\pairPr{\noBP}{\P_{\Airline}}}}{\pairGT{\s}{\G}}$
\item $\derN{\confEl{\Traveler}{\PrConf_1}~  \sep  ~ \confEl{\Hotel}{\PrConf_6} ~  \sep  ~ \confEl{\Airline}{\pairPr{\noBP}{\P_{\Airline}}}}{\pairGT{\s}{\GlComm{\Traveler}{\Airline}\queryM(\infoM).\G_1}}$
\item $\derN{\confEl{\Traveler}{\PrConf_2}~  \sep  ~ \confEl{\Hotel}{\PrConf_6} ~  \sep  ~ \confEl{\Airline}{\PrConf_9}}{\pairGT{\s}{\G_1}}$
\item $\derN{\confEl{\Traveler}{\PrConf_2}~  \sep  ~ \confEl{\Hotel}{\PrConf_7} ~  \sep  ~ \confEl{\Airline}{\PrConf_9}}{\pairGT{\G_1}{\GlComm{\Hotel}{\Traveler}\availableM.\GlComm{\Traveler}{\Airline}\reserveM.\G_2}}$
\item $\derN{\confEl{\Traveler}{\PrConf_3}~  \sep  ~ \confEl{\Hotel}{\PrConf_8} ~  \sep  ~ \confEl{\Airline}{\PrConf_9}}{\pairGT{\G_1}{\GlComm{\Traveler}{\Airline}\reserveM.\G_2}}$
\item $\derN{\confEl{\Traveler}{\PrConf_4}~  \sep  ~ \confEl{\Hotel}{\PrConf_8} ~  \sep  ~ \confEl{\Airline}{\PrConf_{10}}}{\pairGT{\G_1}{\G_2}}$
\item $\derN{\confEl{\Traveler}{\PrConf_4}~  \sep  ~ \confEl{\Hotel}{\PrConf_8} ~  \sep  ~ \confEl{\Airline}{\PrConf_{11}}}{\pairGT{\Push{\G_1}{\G_2}}{\GlComm{\Airline}{\Traveler}\availableM}}$
\item $\derN{\confEl{\Traveler}{\PrConf_5}~  \sep  ~ \confEl{\Hotel}{\PrConf_8} ~  \sep  ~ \confEl{\Airline}{\PrConf_{12}}}{\pairGT{\Push{\G_1}{\G_2}}{\End}}$
\end{myenumerate}
In typing 6, the active type $\G_2$ (see Example~\ref{ex:globalType}) is checkpointed by $B$ and the type of $\Traveler$ is $\proj{\G_2}\Traveler$, see the definition of $\PrConf_4$ in Example~\ref{ex:opSemantics}. In typing 7, the active type is $\GlComm{\Airline}{\Traveler}\availableM$. The type of $\Traveler$ remain the same and without the checkpoint $B$ is a subtype of $\proj{(\GlComm{\Airline}{\Traveler}\availableM)}\Traveler$. In these two typings the type of $\Traveler$ agrees with the first and the second case of condition~\ref{dar5} in Definition~\ref{dar}, respectively.
\end{example}



\mysection{Main Properties}\label{mp}

In this section we present the technical results of the paper. First we prove subject reduction for multiparty sessions (Theorem~\ref{th:subjectReduction}). This implies that well-typed networks respect the choreographies described by global types  (Theorem~\ref{sf}). This property is usually called session fidelity, see~\cite{DY11}. The progress theorem (Theorem~\ref{th:progress}) establishes reachability of all communications and backward reductions. 

As standard we start with an inversion lemma for processes, checkpointed sequences, configurations, multiparty sessions and networks. 

\begin{mylemma}[Inversion]\label{l:inversion}
\begin{myenumerate}{}
\item Let $\derP \Gamma{\P}\T$.
\begin{myenumerate}{}
\item\label{l:inversion1} If $\P=\extChoice{i}{I}{\pp}{\ell}{x}{\P}$, then $\T=\intMsg{i}{I}{\pp}{\ell}{S}{\T}$, and $\derP {\Gamma,\x{:}S_i}{\P_i}\T_i$ for $i\in I$.
\item\label{l:inversion2} If $\P=\intChoice{i}{I}{\pp}{\ell}{\e}{\P}$, then $\T=\unMsg{i}{I}{\pp}{\ell}{S}{\T}$, 
$\derP {\Gamma}{\e_i}S_i$, and $\derP {\Gamma}{\P_i}\T_i$ for $i\in I$. 
\item\label{l:inversion3} If $\P=\CkOne{\extChoice{j}{J}{\pp}{\ell}{x}{\P}}{A}$, then $\T=\CkOne{\intMsg{j}{J}{\pp}{\ell}{S}{\T}}{A}$, and $\derP {\Gamma,\x{:}S_j}{\P_j}\T_j$ for $j\in J$.
\item\label{l:inversion4} If $\P=\CkOne{\intChoice{j}{J}{\pp}{\ell}{\e}{\P}}{A}$, then $\T=\CkTwo{\unMsg{j}{J}{\pp}{\ell}{S}{\T}}{A}$, 
$\derP {\Gamma}{\e_j}S_j$, and $\derP {\Gamma}{\P_j}\T_j$ for $j\in J$. 
\item\label{l:inversion5} If $\P=\mu\recvar.\Q$, then $\derP {\Gamma,\recvar{:}\T}{\Q}\T$.
\item\label{l:inversion6} If $\P=\recvar$, then $\Gamma=\Gamma',\recvar{:}\T$.
\item\label{l:inversion7} If $\P=\inact$,  then $\T=\End$.
\end{myenumerate}
\item \label{l:inversion8} If $\derS{\Push\BP\P}{\RT}$, then $\RT=\PushT{\RT'}\T$ and $\derS{\BP}{\RT'}$ and $\derP{}{\P}{\T}$. 
\item \label{l:inversion9} If $\derN{\confEl{\pp}{\pairPr\BP\P}}{\pairT\RT\T}$, then $\derS{\BP}{\RT}$ and $\derP{}{\P}{\T}$.
\item \label{l:inversion10}  If $\derN{\Pi_{i\in I}\confEl{\pp_i}{\PrConf_i}}{\pairGT{\GG}{\G}}$, then $\derN{\confEl{\pp_i}{\PrConf_i}}{\pairT{\RT_i}{\T_i}}$ and $\agr{\pairT{\RT_i}{\T_i}}{{\pp_i}}{\pairGT{\GG}{\G}}$ for $i\in I$  and\\
 $\len\GG=\text{max}\set{\len{\RT_i}\mid i\in I}$ and  $\pt(\GG)\cup\pt(\G)\subseteq\set{\pp_i\mid i\in I}$. 
\item \label{l:inversion11} If $\derNok{\Nt}$, then either $\derN{\Nt}{\pairGT{\GG}{\G}}$ or $\Nt=\Nt'\pc\Nt''$ and $\derNok{\Nt'}$ and $\derNok{\Nt''}$. 
\end{myenumerate}
\end{mylemma}
\begin{proof}
Easy from the definition of the typing relation.
\end{proof}
The inversion lemma gives some important properties. 

Points~(\ref{l:inversion1}), (\ref{l:inversion2}), (\ref{l:inversion3}),(\ref{l:inversion4})  and (\ref{l:inversion7}) ensure that the processes are checkpointed iff their types are checkpointed with the same name. Moreover, input processes have intersection types, output processes have union types and the process $\inact$ has type $\End$. 

Point~(\ref{l:inversion8}) says that the length of checkpointed sequences is equal to the length of the sequences of their checkpointed session types. 

Point~(\ref{l:inversion10}) and Definition~\ref{dar} imply that in a well-typed multiparty session:
\begin{myitemize}
\item exactly one of the active processes  is an output process;
\item at least one of the active processes  is an input process.
\end{myitemize}
More precisely, if $\G=\CkGl{\Gvtj\pp \ell \ST {\G}}{A}$, then the active processes of participants $\pp$ and $\q$ must have types which are subtypes of $\proj\G\pp$ and $\proj\G\q$, respectively. E.g., this is the case of
typing 3 of Example~\ref{ex:typing}, where the active processes of $\Hotel$ and $\Traveler$ have types equal to the projections the active global type onto the respective participant. (Similarly for typing 6.)
If $\G=\Gvti\pp \ell \ST {\G}$, then the active process of participant $\pp$ must have a type which is a subtype of $\proj\G\pp$, as before. Instead, the active process of participant $\q$ can have either a type which is a subtype of $\proj\G\q$, or a checkpointed type $\CkOne{\T}A$ such that $\T\subt\proj\G\q$. 
E.g., this is the case of
typing 7 in Example~\ref{ex:typing}, where the active processes of $\Airline$ has an uncheckpointed 
output type (a union of a single type), whereas the active process of $\Traveler$  
has an intersection type checkpointed by $B$.
If $\G$ is checkpointed by $A$, then an active process different from $\inact$ is checkpointed by $A$ and the checkpointed sequence  has length $\len\GG$. Instead, if  $\G$ is uncheckpointed, then the output process is uncheckpointed and its checkpointed sequence  has length $\len\GG$, while an input process can be:
\begin{myitemize}
\item either uncheckpointed: in this case its checkpointed sequence  has length $\len\GG$, 
\item or checkpointed: in this case its checkpointed sequence  has length $\len\GG-1$.
\end{myitemize}
E.g., for the case of
typing 3 of Example~\ref{ex:typing}, the length of the checkpointed sequences of all participants is 0, and all
participants have checkpointed active processes, whereas
for
typing 7, the length of the checkpointed sequence of $\Airline$ is 2, whereas the ones of the other participants is 1.
Moreover, the active process of $\Traveler$ is checkpointed and the one of $\Hotel$ is $\inact$.
These properties follow from conditions~\ref{dar4} and~\ref{dar5} of Definition~\ref{dar}. Condition~\ref{dar5} of Definition~\ref{dar} implies also that, if $\G$ is uncheckpointed and the type of an active input process is checkpointed, then the name of its checkpoint is the name of the checkpoint of the global type on the top of $\GG$. Lastly, condition~\ref{dar1} of Definition~\ref{dar} ensures that all the processes in the checkpointed sequences are checkpointed by the same names as the global types in $\GG$, in the exact order. The only difference can be the length of the sequences, which must satisfy the other conditions of Definition~\ref{dar}.

Global types are not preserved under multiparty session reductions: this is expected, as they evolve according to the silent actions, the communications and  the rollbacks performed by the session participants. This evolution is formalised by the {\em reduction of global types}, which is the smallest pre-order relation closed under the rules of  Figure~\ref{redGlConf}. 
Rule \rln{[G-\ck\chc]} corresponds to the reduction of the output process of participant $\pp$ by rule \rln{[\ck\chc]} of Figure~\ref{rrc}. Notably, the process is checkpointed with name $A$. 
Rule \rln{[G-Com]} corresponds to the communication performed by rule \rln{[Com]} of Figure~\ref{rrsn}. I.e., the output process of participant $\pp$ sends a message labelled $\ell_k$ (rule \rln{[\snd]} of Figure~\ref{rrc}) and the input process of participant $\q$ receives this message  (rule \rln{[Rcv]} or \rln{[\ck Rcv]} of Figure~\ref{rrc}). Notice that, when rule \rln{[\ck Rcv]} is used, the checkpointed input process is added to the checkpointed sequence of participant $\q$.  Lastly, rule \rln{[G-\rb]} is used when the multiparty session rolls back by means of rule \rln{[\rb M]} of Figure~\ref{rrsn}. Let $A$ be the name of the checkpoint of $\G$, the participants must either reduce by  rule \rln{[\rb P]} of Figure~\ref{rrc} with a transition labelled $A$ or remain unchanged. In the unchanged configurations the processes belonging to the checkpointed sequences are not checkpointed by $A$ and  the active processes are $\inact$. We use $\redGstar$ to denote the transitive
and reflexive closure of the $\redG$ relation.

\begin{figure}[h] 
{\small \begin{center}
$
\begin{array}{l} 
\pairGT{\GG}{\CkGl{\Gvtj\pp \ell \ST {\G}}{A}}\redG
{\Push{\GG}
{(\CkGl{\GvtiLong{\pp}{\q}{\ell_j(\ST_j).\G_j}_{j\in J}}{A})}
}\prec\,{\GvtiLong{\pp}{\q}{\ell_j(\ST_j).\G_j}_{j\in J}}\quad\quad\rln{[G-\ck\chc]}\\ \\
\myruleOneLine{\pairGT{\GG}{\Gvti\pp \ell \ST {\G}}\redG
\pairGT{\GG}{\G_k}\quad k\in I}{[G-Com]}\\ \\
\myruleOneLine{\pairGT{\Push{\Push\GG\G}{\GG'}}{\G'}\redG
\pairGT\GG\G}{[G-\rb]}
\end{array}
$
\end{center}
}
\caption{Reduction rules of global types.
}\labelx{redGlConf}
\end{figure}

A standard substitution lemma is handy.
\begin{lemma}\label{s}
If $\derP {\Gamma,\x:\ST}{\P}{\T}$ and $\dere {\Gamma}{\e}{\ST}$ and $\eval\e\val$, then $\derP {\Gamma}{\subst\P\x\val}{\T}$.
\end{lemma}

We can show subject reduction for well-typed multiparty sessions, which  implies subject reduction for well-typed networks. 
\begin{theorem}[SR]\label{th:subjectReduction}
If $\derN{\Nw}{\pairGT\GG\G}$ and 
$\Nw\stackredstar{~~\silent~~}\Nw'$, then $\derN{\Nw'}{\pairGT{\GG'}{\G'}}$ and  $\pairGT{\GG}{\G}\redGstar\pairGT{\GG'}{\G'}$.
\end{theorem}
\begin{proof}
By induction on multiparty session reductions. It is easy to verify that typing is invariant under structural equivalence of multiparty sessions, so we will omit the application of rule \rln{[EqM]}.\\
 If $\Nw\stackred{~~\silent~~}\Nw'$, then there are three cases:
\begin{myenumerate}{}
\item\label{c1} $\Nw=\confEl{\pp}{\PrConf}\pc\Nw''$ and $\Nw'=\confEl{\pp}{\PrConf'}\pc\Nw''$ and $\confEl{\pp}{\PrConf}\stackred{\silent}\confEl{\pp}{\PrConf'}$, i.e., rule \rln{[PrM]} has been applied, 
\item\label{c2} $\Nw=\confEl{\pp}{\PrConf_\pp}\pc\confEl{\q}{\PrConf_\q}\pc\Nw''$ and $\Nw'=\confEl{\pp}{\PrConf_\pp'}\pc\confEl{\q}{\PrConf_\q'}\pc\Nw''$ and $\confEl{\pp}{\PrConf_\pp}\stackred{\q!\ell(\val)}\confEl{\pp}{\PrConf'_\pp}$ and\\ $\confEl{\q}{\PrConf_\q}\stackred{\pp?\ell(\val)}\confEl{\q}{\PrConf_\q'}$, i.e., rule \rln{[PrM]}, with rule \rln{[Com]} on the premise, has been applied, 
\item\label{c3} $\Nw=\Pi_{i\in I}{\confEl{\pp_i}{\PrConf_{\pp_i}}\pc\Pi_{j\in J}\confEl{\pp_j}{\PrConf_{\pp_j}}}$ and $\Nw'=\Pi_{i\in I}\confEl{\pp_i}{\PrConf_{\pp_i}'}\pc\Pi_{j\in J}\confEl{\pp_j}{\PrConf_{\pp_j}}$ and 
$\confEl{\pp_i}{\PrConf_{\pp_i}}\stackred{A}\confEl{\pp_i}{\PrConf'_{\pp_i}}$ for all $i\in I$ and $A\not\in\A(\PrConf_{\pp_j})$ for all $j\in J$, i.e., rule \rln{[\rb M]} has been applied.
\end{myenumerate}

\noindent
\underline{Case (\ref{c1})}.  From $\confEl{\pp}{\PrConf}\stackred{\silent}\confEl{\pp}{\PrConf'}$ we get ${\PrConf}\stackred{\silent}{\PrConf'}$. Let $\PrConf=\pairPr{\BP}\P$ and  $\PrConf'=\pairPr{\BP'}{\P'}$. Therefore \vspace{-2pt}
\begin{enumerate}[(a)]
\item\label{a} either $\PrConf\stackred{\silent}\PrConf'$ with rule \rln{[\chc]}, which implies $\P=\intChoice{i}{I}{\q}{\ell}{\e}{\P}$ and $\BP'=\BP$ and $\P'=\sendMsg{\q}{\ell_k}{\e_k}{\P_k}$ for $k\in I\not=\set k$,
\item\label{b} or $\PrConf\stackred{\silent}\PrConf'$ with rule \rln{[\ck\chc]}, which implies $\P=\CkOne{\intChoice{j}{J}{\q}{\ell}{\e}{\P}}{A}$ and $\BP'=\Push{\BP}{\P}$ and $\P'=\sendMsg{\q}{\ell_k}{\e_k}{\P_k}$ for $k\in J\neq\{k\}$.
\end{enumerate}\vspace{-2pt}
By Lemma~\ref{l:inversion}(\ref{l:inversion10}) $\derN{\confEl{\pp}{\PrConf}}{\pairT{\RT}{\T}}$ and $\agr{\pairT{\RT}{\T}}{{\pp}}{\pairGT{\GG}{\G}}$. 
Then $\derS{\BP}{\RT}$ and $\derP{}{\P}{\T}$ by Lemma~\ref{l:inversion}(\ref{l:inversion9}).  Lemma~\ref{l:inversion}(\ref{l:inversion2})  and (\ref{l:inversion4}) imply that $\T$ is a union type, then we get $\len\RT=\len\GG$ and $\T\subt\proj\G\pp$ by condition~\ref{dar4} of Definition~\ref{dar}. 

\noindent
\underline{Case (\ref{c1}\ref{a})}. Lemma~\ref{l:inversion}(\ref{l:inversion2}) applied to $\derP{}{\P}{\T}$ gives $\T=\unMsg{i}{I}{\q}{\ell}{\ST}{\T}$ and $\dere {}{\e_i}{\ST_i}$ and $\derP {}{\P_i}{\T_i}$ for $i\in I$. 
Then $\T\subt\proj\G\pp$ implies $\G=\GvtiLong{\pp}{\q}{\ell_i(\ST_i).\G_i}_{i\in I\cup L}$. We choose $\GG'=\GG$ and $\G'=\G$. In fact, we can derive $\derP {}{\P'}{\q!\ell_k(\ST_k).\T_k}$ and $\q!\ell_k(\ST_k).\T_k\subt\proj\G\pp$. Therefore, we derive $\derN{\Nw'}{\pairGT{\GG'}{\G'}}$ by rule \rln{[t-M]}.

\noindent
\underline{Case (\ref{c1}\ref{b})}. Lemma~\ref{l:inversion}(\ref{l:inversion4})  applied to $\derP{}{\P}{\T}$ gives $\T=\CkOne{\unMsg{j}{J}{\q}{\ell}{S}{\T}}A$ and $\dere {}{\e_j}{\ST_j}$ and $\derP {}{\P_j}{\T_j}$ for $j\in J$. Then $\T\subt\proj\G\pp$ implies $\G=\CkOne{\GvtiLong{\pp}{\q}{\ell_j(\ST_j).\G_j}_{j\in J\cup L}}A$. We can choose $\GG'=\Push\GG\G$ and $\G'=\GvtiLong{\pp}{\q}{\ell_j(\ST_j).\G_j}_{j\in J\cup L}$. In fact $\pairGT{\GG}{\G}\redG\pairGT{\Push\GG\G}{\G'}$ by rule \rln{[G-\ck\chc]} and we can derive $\derP {}{\P'}{\q!\ell_k(\ST_k).\T_k}$ and $\q!\ell_k(\ST_k).\T_k\subt\proj{\G'}\pp$.\\ If $\pr\not=\pp$ and $\confEl{\pr}{\pairPr{\BP_\pr}\P_\pr}$ occurs in $\Nw$, then by Lemma~\ref{l:inversion}(\ref{l:inversion10}) $\derN{\confEl{\pr}{\pairPr{\BP_\pr}\P_\pr}}{\pairT{\RT_\pr}{\T_\pr}}$ and $\agr{\pairT{\RT_\pr}{\T_\pr}}{{\pr}}{\pairGT{\GG}{\G}}$. Lemma~\ref{l:inversion}(\ref{l:inversion9}) gives $\derS{\BP_\pr}{\RT_\pr}$ and $\derP{}{\P_\pr}{\T_\pr}$. If $\P_\pr=\inact$, then $\T_\pr=\proj\G\pr=\proj{\G'}\pr=\End$ and \mbox{$\agr{\pairT{\RT_\pr}{\T_\pr}}{{\pr}}{\pairGT{\GG'}{\G'}}$,} since 
condition~\ref{dar3} of
Definition~\ref{dar} is satisfied. If $\P_\pr$ is an input process, then $\T_\pr$ is an intersection type and $\len{\RT_\pr}=\len{\GG}$ and $\T_\pr\subt\proj\G\pr$ by  the first alternative in condition~\ref{dar5} of Definition~\ref{dar}. We have $\len{\RT_\pr}=\len{\GG'}-1$ since $\len{\RT_\pr}=\len{\GG}$. From  $\T_\pr\subt\proj\G\pr$ we get $\T_\pr= \CkOne{\T_\pr'}A$ and $\T_\pr'\subt\proj{\G'}\pr$. Then $\agr{\pairT{\RT_\pr}{\T_\pr}}{{\pr}}{\pairGT{\GG'}{\G'}}$ since the second alternative of condition~\ref{dar5} in Definition~\ref{dar} is satisfied. We can then derive $\derN{\Nw'}{\pairGT{\GG'}{\G'}}$ by rule \rln{[t-M]}.

\noindent
\underline{Case (\ref{c2})}. From $\confEl{\pp}{\PrConf_\pp}\stackred{\q!\ell(\val)}\confEl{\pp}{\PrConf'_\pp}$ and $\confEl{\q}{\PrConf_\q}\stackred{\pp?\ell(\val)}\confEl{\q}{\PrConf_\q'}$ we get that $\PrConf_\pp\stackred{\q!\ell(\val)}\PrConf'_\pp$ and $\PrConf_\q\stackred{\pp?\ell(\val)}\PrConf_\q'$. Let $\PrConf_\pp=\pairPr{\BP_\pp}{\P_\pp}$ and  $\PrConf'_\pp=\pairPr{\BP'_\pp}{\P'_\pp}$ and $\PrConf_\q=\pairPr{\BP_\q}\P_\q$ and  $\PrConf'_\q=\pairPr{\BP'_\q}{\P'_\q}$. Then $\PrConf_\pp$ reduces with rule \rln{[\snd]}, which implies $\P_\pp=\sendMsg{\q}{\ell}{\e}{\P}$ and $\eval{\e}{\val}$ and $\BP'_\pp=\BP_\pp$ and $\P'_\pp=\P$. The reduction of $\PrConf_\q$ can be done \vspace{-10pt}
\begin{enumerate}[(a)]
\item\label{2a} either with rule  \rln{[Rcv]}, which implies $\P_\q=\extChoice{i}{I}{\pp}{\ell}{x}{\P}$ with $\ell_k=\ell$ and $\BP'_\q=\BP_\q$ and $\P'_\q=\subst{\P_k}{x}{\val}$, 
\item\label{2b} or with rule \rln{[CkRcv]}, which implies $\P_\q=\CkOne{\extChoice{j}{J}{\pp}{\ell}{x}{\P}}{A}$ with $\ell_k=\ell$ and $\BP'_\q=\Push{\BP_\q}{\P_\q}$ and $\P'_\q=\subst{\P_k}{x}{\val}$.
\end{enumerate}\vspace{-3pt}
By Lemma~\ref{l:inversion}(\ref{l:inversion10}) $\derN{\confEl{\pp}{\PrConf_\pp}}{\pairT{\RT_\pp}{\T_\pp}}$ and $\agr{\pairT{\RT_\pp}{\T_\pp}}{{\pp}}{\pairGT{\GG}{\G}}$ and $\derN{\confEl{\q}{\PrConf_\q}}{\pairT{\RT_\q}{\T_\pp}}$ and $\agr{\pairT{\RT_\q}{\T_\q}}{{\q}}{\pairGT{\GG}{\G}}$. 
Then $\derS{\BP_\pp}{\RT_\pp}$ and $\derP{}{\P_\pp}{\T_\pp}$ and $\derS{\BP_\q}{\RT_\q}$ and $\derP{}{\P_\q}{\T_\q}$ by Lemma~\ref{l:inversion}(\ref{l:inversion9}).  Lemma~\ref{l:inversion}(\ref{l:inversion2}) applied to $\derP{}{\P_\pp}{\T_\pp}$ gives $\T_\pp=\tout{\q}{\ell}{\ST}.{\T}$ and $\dere {}{\e}{\ST}$ and $\derP {}{\P}{\T}$. We have $\len{\RT_\pp}=\len\GG$ and $\T_\pp\subt\proj\G\pp$ by condition~\ref{dar4} of Definition~\ref{dar}. This implies $\G=\Gvth\pp \ell \ST {\G}$ with $\ell_k=\ell$, $\ST=\ST_k$ and $\T\subt\proj{\G_k}\pp$. We derive $\derN{\confEl{\pp}{\PrConf_\pp'}}{\pairT{\RT_\pp}{\T}}$. 

\noindent
We can choose $\GG'=\GG$ and $\G'=\G_k$ since $\pairGT{\GG}{\G}\redG\pairGT{\GG}{\G_k}$ by rule \rln{[G-Com]} and we will show that $\derN{\Nw'}{\pairGT{\GG'}{\G'}}$ is derivable by checking the agreement conditions of Definition~\ref{dar} for all pairs participant/configuration of $\Nw'$. From $\agr{\pairT{\RT_\pp}{\T_\pp}}{{\pp}}{\pairGT{\GG}{\G}}$ and $\T\subt\proj{\G_k}\pp$ we get $\agr{\pairT{\RT_\pp}{\T}}{{\pp}}{\pairGT{\GG'}{\G'}}$.

\noindent
\underline{Case (\ref{c2}\ref{2a})}. Lemma~\ref{l:inversion}(\ref{l:inversion1})  applied to $\derP{}{\P_\q}{\T_\q}$ gives $\T_\q=\intMsg{i}{I}{\q}{\ell}{S}{\T}$ and  $\derP {x_i:\ST_i}{\P_i}{\T_i}$ for $i\in I$. We get $\len{\RT_\pp}=\len\GG$ and $\T_\q\subt\proj\G\q$ by condition~\ref{dar5} of Definition~\ref{dar}. This implies $H\subseteq I$ and in particular $k\in I$. The Substitution Lemma implies $\derP{}{\P'_\q}{\T_k}$. We derive $\derN{\confEl{\q}{\PrConf_\q'}}{\pairT{\RT_\q}{\T_k}}$. From $\agr{\pairT{\RT_\q}{\T_\q}}{{\q}}{\pairGT{\GG}{\G}}$ and $\T_k\subt\proj{\G_k}\q$ we get $\agr{\pairT{\RT_\q}{\T_k}}{{\q}}{\pairGT{\GG'}{\G'}}$.

\noindent
\underline{Case (\ref{c2}\ref{2b})}. Lemma~\ref{l:inversion}(\ref{l:inversion3})  applied to $\derP{}{\P_\q}{\T_\q}$ gives $\T_\q=\CkOne{\intMsg{j}{J}{\q}{\ell}{S}{\T}}A$ and  $\derP {x_j:\ST_j}{\P_j}{\T_j}$ for $j\in J$. Let $\GG=\Push{\GG''}{\G''}$. We get $\len{\RT_\pp}=\len\GG-1$ and $\T_\q\subt\proj{\G''}\q$ and $\intMsg{j}{J}{\q}{\ell}{S}{\T}\subt\proj\G\q$ by condition~\ref{dar5} of Definition~\ref{dar}. This implies $H\subseteq J$ and in particular $k\in J$. As in previous case the Substitution Lemma implies $\derP{}{\P'_\q}{\T_k}$. We derive $\derN{\confEl{\q}{\PrConf_\q'}}{\pairT{\Push{\RT_\q}{\T_\q}}{\T_k}}$. From $\agr{\pairT{\RT_\q}{\T_\q}}{{\q}}{\pairGT{\Push{\GG''}{\G''}}{\G}}$  and $\T_\q\subt\proj{\G''}\q$ and $\T_k\subt\proj{\G_k}\q$ we get $\agr{\pairT{\Push{\RT_\q}{\T_\q}}{\T_k}}{{\q}}{\pairGT{\GG'}{\G'}}$.

\noindent
Consider a participant $\pr\not=\pp,\q$. If $\confEl{\pr}{\pairPr{\BP_\pr}\P_\pr}$ occurs in $\Nw$, then by Lemma~\ref{l:inversion}(\ref{l:inversion10}) $\derN{\confEl{\pr}{\pairPr{\BP_\pr}\P_\pr}}{\pairT{\RT_\pr}{\T_\pr}}$ and $\agr{\pairT{\RT_\pr}{\T_\pr}}{{\pr}}{\pairGT{\GG}{\G}}$. Lemma~\ref{l:inversion}(\ref{l:inversion9}) gives $\derS{\BP_\pr}{\RT_\pr}$ and $\derP{}{\P_\pr}{\T_\pr}$. If $\P_\pr=\inact$, then $\T_\pr=\proj\G\pr=\proj{\G'}\pr=\End$ and $\agr{\pairT{\RT_\pr}{\T_\pr}}{{\pr}}{\pairGT{\GG'}{\G'}}$, since condition~\ref{dar3} of Definition~\ref{dar} is satisfied. If $\P_\pr$ is an input process, then $\T_\pr$ is an intersection type and either $\len{\RT_\pr}=\len{\GG}$ and $\T_\pr\subt\proj\G\pr$ or $\len{\RT_\pr}=\len{\GG}-1$ and $\T_\pr\subt\proj{\G''}\pr$ and $\T_\pr=\CkOne{\T'}A$ and $\T'\subt\proj\G\pr$ by condition~\ref{dar5} of Definition~\ref{dar}. In both cases $\proj\G\pr\subt\proj{\G_k}\pr$. If  $\G_k$ is uncheckpointed we conclude $\agr{\pairT{\RT_\pr}{\T_\pr}}{{\pr}}{\pairGT{\GG'}{\G'}}$. 
If $\G_k$ is checkpointed we must have $\len{\RT_\pr}=\len{\GG}$ and $\T_\pr\subt\proj\G\pr$. In fact otherwise $\T'\subt\proj{\G_k}\pr$ would imply $\T'=\CkOne{\T''}B$, where $B$ is the name of the checkpoint of $\G_k$. We would get $\T_\pr=\CkOne{\CkOne{\T''}B}A$ and this is not a session type according to our syntax.

\noindent
\underline{Case (\ref{c3})}. Let $\GG=\Push{\Push{\GG'}{\G'}}{\GG''}$ and $A$ be the name of the checkpoint of $\G'$. Lemma~\ref{l:inversion}(\ref{l:inversion10}) implies that \mbox{$\derN{\confEl{\pp_l}{\PrConf_{\pp_l}}}{\pairT{\RT_{\pp_l}}{\T_{\pp_l}}}$} and $\agr{\pairT{\RT_{\pp_l}}{\T_{\pp_l}}}{{\pp_l}}{\pairGT{\GG}{\G}}$ for all $l\in I\cup J$.  If $\len{\RT_{\pp_l}}\leq\len{\GG'}$, then
$A\not\in\AA(\PrConf_{\pp_l})$ by Definition~\ref{sgtypesdef}(\ref{sgtypesdef2}) and Definition~\ref{dar}. This implies $l\in J$ and $\agr{\pairT{\RT_{\pp_l}}{\T_{\pp_l}}}{{\pp_l}}{\pairGT{\GG'}{\G'}}$. Otherwise $l\in I$ and $\RT_{\pp_l}=\Push{\Push{\RT'_{\pp_l}}{\T_l}}{\RT''_{\pp_l}}$ and $\T_l\subt\proj{\G'}{\pp_l}$ by condition~\ref{dar1} of Definition~\ref{dar}. Let $\PrConf_{\pp_i}=\pairPr{\BP_{\pp_i}}{\P_{\pp_i}}$. From $\derN{\confEl{\pp_i}{\PrConf_{\pp_i}}}{\pairT{\RT_{\pp_i}}{\T_{\pp_i}}}$ we get $\derN{\BP_{\pp_i}}{\Push{\Push{\RT'_{\pp_i}}{\T_i}}{\RT''_{\pp_i}}}$ by Lemma~\ref{l:inversion}(\ref{l:inversion9}). Then $\BP_{\pp_i}=\Push{\Push{\BP'_{\pp_i}}{\P_i}}{\BP''_{\pp_i}}$ and $\derN{\P_i}{\T_i}$ by Lemma~\ref{l:inversion}(\ref{l:inversion8}). The name of the checkpoint of $\P_i$ is $A$ since $\T_i\subt\proj{\G'}{\pp_i}$ and then $\PrConf_{\pp_l}'=\pairPr{\BP'_{\pp_i}}{\P_i}$. This implies $\agr{\pairT{\RT_{\pp_i}}{\T_{\pp_i}}}{{\pp_i}}{\pairGT{\GG'}{\G'}}$. We can then derive $\derN{\Nw'}{\pairGT{\GG'}{\G'}}$ by rule \rln{[t-M]}.

\end{proof}

From the proof of the previous theorem we get the following properties of well-typed networks, which are usual for session calculi~\cite{CHY08}. We say that an application of the reduction rule \rln{[Com]} has a {\em type mismatch}, if there is no sort that can be derived both for the communicated value and for variable associated to the communicated label in the input process. 

Global types describe interaction protocols. The communications in well-typed networks evolve following the exact order of the associated global types. 

\begin{theorem}[Session Fidelity]\label{sf}
If $\derNok{\Nt}$, then reducing $\Nt$
\begin{myenumerate}{}
\item\label{sf1} there is never a type mismatch;
\item\label{sf2} the communications occur in the order prescribed by global types. 
\end{myenumerate}
\end{theorem}
Notably property~\ref{sf1} holds in spite of the fact that session participants may exchange messages of  different types. Property~\ref{sf2} says that session participants behave according to established communication protocols.

The standard definition of progress only ensures absence of deadlocks~\cite[Section 8.3]{pier02}. Progress for session calculi means that all the requested interactions may happen~\cite{CDPY16}. In reversible sessions it is also natural to guarantee that all possible rollbacks may take place. This leads us to the following formulation of the progress theorem.

\begin{theorem}[Progress]\label{th:progress}
If $\derNok{\Nt}$, then:
\begin{myenumerate}{}
\item\label{th:progress1} if $\Nt$ contains an input or output process, then $\Nt$ forward reduces to $\Nt'$ and that input or output prefix does not occur in $\Nt'$;
\item\label{th:progress2} if $\Nt$ contains a checkpoint named $A$, then there is a reduction of $\Nt$ in which the last step is a rollback making the processes checkpointed by $A$ active processes. 
\end{myenumerate}
\end{theorem}
\begin{proof} As proved in \cite{CDPY16}, a single multiparty session
  in a standard calculus with global and session types, like the
  calculus in~\cite{CHY08}, always enjoys progress whenever it is well typed. 
  In fact, by the Subject Reduction Theorem (Theorem~\ref{th:subjectReduction}), reduction preserves
  well-typedness of sessions. Moreover, all required session participants are
  present, as ensured by the condition $\pt(\GG)\cup\pt(\G)\subseteq\set{\pp_i\mid i\in I}$ in the premise of rule \rln{[t-M]}. Thus, all communications among participants in a unique
  session will take place, in the order prescribed by the single-threaded active global type. This ensures that property~\ref{th:progress1} holds. For property~\ref{th:progress2} observe that, 
  if $\Nt$ contains checkpoints named $A$, then there is at least one multiparty session $\Nw$ in $\Nt$ which is typed by a global type $\pairGT\GG\G$ which contains $A$. If $A$ occurs in $\G$, then we can reduce forward $\Nw$ until the processes checkpointed by $A$ will be all in the checkpointed sequences, and then apply the desired rollback. If $A$ does not occur in $\G$ let $\GG=\Push{\Push{\GG'}{\CkOne{\G'}B}}{\GG''}$ and $A$ occurs in $\CkOne{\G'}B$.  Then the checkpointed sequences of $\Nw$ have processes checkpointed by $B$. We can then apply the rollback which makes the processes checkpointed by $B$ to become active processes. If $A=B$ we are done. Otherwise the active  global type of the obtained session contains $A$ and we can conclude as in previous case. 
\end{proof}


\mysection{Related Work and Conclusions}\label{rwc}

Since the pioneering work by Danos and Krivine~\cite{DK04}, reversible computations in process algebras have been widely studied. The calculus of~\cite{DK04} adds a distributed monitoring system to CCS~\cite{M89} allowing computations to be rewound. 
Phillips and Ulidowski~\cite{PU07} propose a method for reversing process operators that are definable by SOS rules in a general format, using keys to bind synchronised actions together.
A reversible variant of the higher-order $\pi$-calculus is defined in~\cite{LMS10}, using name tags for identifying threads and explicit memory processes. In~\cite{LMSS11}, Lanese et {al.} enrich the calculus of~\cite{LMS10} with a fine-grained rollback primitive. 
To the best of our knowledge the earliest work dealing with rollback of communicating systems
are~\cite{VriesKH10,VKH10b,KSH14}. In these papers an extension of CCS
models the combination of rollback recovery and coordinated
checkpoints.  

As pointed out in~\cite{PU07}, reversibility in process calculi is challenging, since  we cannot distinguish between the processes $a\| a$ and $a.a$ by simply recording the past actions. For this reason both histories and unique identifiers for threads have been used to track information. We do not have this problem in our calculus since each session participant reduces in a sequential way. A key requirement, dubbed {\em causal consistency} in ~\cite{DK04}, is that of undoing actions only if no other action depending on them has been executed (and not undone). In the present work causal consistency follows from the linearity of the interactions described by single-threaded global types. 

The most widely used models of structured communication-based programming are session behaviours \cite{BdL13,BH13} and session calculi~\cite{HVK98,CHY08}. Reversibility has been incorporated into both these models. 

Compliance and sub-behaviour for session behaviours with checkpoints has been first studied in~\cite{BDL16}. There a process has the possibility, after a rollback, of resuming the computation along the very 
same branch of the computation on which the rollback has been performed. From a different point of view, instead, rollbacks could be used as a strategy to get compliance.
For instance assuming the interacting processes to roll back whenever the current branch of
the computation cannot proceed and a different branch could work instead.
This approach has been investigated in~\cite{BDLdL15}.

The papers closer to ours are~\cite{TY15,TY16,MP16}. Tiezzi and Yoshida~\cite{TY15} use tags and memories to allow reversibility of binary sessions with delegation. Reversibility is full, i.e. each interaction can be undone and causal consistency is preserved. An extension of this calculus allows computation to go forward and backward until the session is committed by means of a specific irreversible action. Only processes are typed, but this is enough to ensure absence of errors. Two forms of reversibility are considered in~\cite{TY16}. Either a session can be completely reversed with one backward step, or any intermediate state can be restored with either one backward step or multiple ones.  In the first case the memory is just the initial process, while in the second case the sequence of all the processes generated by the reduction is needed. Both binary and multiparty sessions are taken into account under the hypothesis that they are ``single''. A session is single when all participants interact only along that session. Mezzina and P\'erez~\cite{MP16} use monitors as memories. A key novelty of~\cite{MP16} are session types with present and past, which allow the semantics of reversible actions to be streamlined.



The main contributions of this paper are the treatment of checkpointed interactions and the role played by global types in controlling reversibility.
In defining of our calculus, we made some design choices. For simplicity, we did not consider 
\begin{myitemize}
\item session initialisation by means of request/accept,
\item subsorting and covariance/contravariance of messages types in the subtyping of session types,
\item asynchronous communications using message queues.
\end{myitemize}
Including these features, which are present in~\cite{CHY08,GH05,TY16}, would be easy.\\
Moreover we did the following assumptions:
\begin{myitemize}
\item the rollback to a checkpoint is done non-deterministically and simultaneously by all participants which traversed that checkpoint,
\item all the communications can be undone. 
\end{myitemize}
In future work we plan to address the issue of communication that cannot be undone, such as ``money dispensed by an ATM machine'', and also add to the process language primitives triggering the rollback. 
In our calculus, when crossing a checkpoint we memorise all the branches of the choice. Including 
only the branches not taken, would get us also checkpointed single inputs/outputs, 
and this would require some care.

We will also study rollbacks with checkpoints for interleaved multiparty sessions with delegation. In this case, a crucial point is the dependency between different sessions when backward reductions are done, see~\cite{TY15}. 


\vspace{-15pt}
\paragraph{Acknowledgments.}
We are grateful to the anonymous reviewers for their useful
suggestions, which led to substantial improvements. 
\vspace{-15pt}


\begin{thebibliography}{10}
\providecommand{\bibitemdeclare}[2]{}
\providecommand{\urlprefix}{Available at }
\providecommand{\url}[1]{\texttt{#1}}
\providecommand{\href}[2]{\texttt{#2}}
\providecommand{\urlalt}[2]{\href{#1}{#2}}
\providecommand{\doi}[1]{doi:\urlalt{http://dx.doi.org/#1}{#1}}
\providecommand{\bibinfo}[2]{#2}

\bibitemdeclare{article}{BDL16}
\bibitem{BDL16}
\bibinfo{author}{Franco Barbanera}, \bibinfo{author}{Mariangiola
  Dezani-Ciancaglini} \& \bibinfo{author}{Ugo de'Liguoro}
  (\bibinfo{year}{2016}): \emph{\bibinfo{title}{Reversible client/server
  interactions}}.
\newblock {\sl \bibinfo{journal}{Formal Aspects of Computing}}
  \bibinfo{volume}{28}(\bibinfo{number}{4}), pp. \bibinfo{pages}{697--722},
  \doi{10.1007/s00165-016-0358-2}.

\bibitemdeclare{inproceedings}{BDLdL15}
\bibitem{BDLdL15}
\bibinfo{author}{Franco Barbanera}, \bibinfo{author}{Mariangiola
  Dezani{-}Ciancaglini}, \bibinfo{author}{Ivan Lanese} \& \bibinfo{author}{Ugo
  de' Liguoro} (\bibinfo{year}{2016}): \emph{\bibinfo{title}{Retractable
  Contracts}}.
\newblock In: {\sl \bibinfo{booktitle}{PLACES}}, {\sl
  \bibinfo{series}{{EPTCS}}} \bibinfo{volume}{203}, pp.
  \bibinfo{pages}{61--72}, \doi{10.4204/EPTCS.203}.

\bibitemdeclare{article}{BdL13}
\bibitem{BdL13}
\bibinfo{author}{Franco Barbanera} \& \bibinfo{author}{Ugo de' Liguoro}
  (\bibinfo{year}{2015}): \emph{\bibinfo{title}{{Sub-behaviour relations for
  session-based client/server systems}}}.
\newblock {\sl \bibinfo{journal}{Mathematical Structures in Computer Science}}
  \bibinfo{volume}{25}(\bibinfo{number}{6}), pp. \bibinfo{pages}{1339--1381},
  \doi{10.1017/S096012951400005X}.

\bibitemdeclare{article}{BH13}
\bibitem{BH13}
\bibinfo{author}{Giovanni Bernardi} \& \bibinfo{author}{Matthew Hennessy}
  (\bibinfo{year}{2016}): \emph{\bibinfo{title}{Modelling session types using
  contracts}}.
\newblock {\sl \bibinfo{journal}{Mathematical Structures in Computer Science}}
  \bibinfo{volume}{26}(\bibinfo{number}{3}), pp. \bibinfo{pages}{510--560},
  \doi{10.1017/S0960129514000243}.

\bibitemdeclare{article}{CDPY16}
\bibitem{CDPY16}
\bibinfo{author}{Mario Coppo}, \bibinfo{author}{Mariangiola
  Dezani-Ciancaglini}, \bibinfo{author}{Nobuko Yoshida} \&
  \bibinfo{author}{Luca Padovani} (\bibinfo{year}{2016}):
  \emph{\bibinfo{title}{Global Progress for Dynamically Interleaved Multiparty
  Sessions}}.
\newblock {\sl \bibinfo{journal}{Mathematical Structures in Computer Science}}
  \bibinfo{volume}{26}(\bibinfo{number}{2}), pp. \bibinfo{pages}{238--302},
  \doi{10.1017/S0960129514000188}.

\bibitemdeclare{inproceedings}{DK04}
\bibitem{DK04}
\bibinfo{author}{Vincent Danos} \& \bibinfo{author}{Jean Krivine}
  (\bibinfo{year}{2004}): \emph{\bibinfo{title}{Reversible Communicating
  Systems}}.
\newblock In: {\sl \bibinfo{booktitle}{CONCUR}}, {\sl \bibinfo{series}{LNCS}}
  \bibinfo{volume}{3170}, \bibinfo{publisher}{Springer}, pp.
  \bibinfo{pages}{292--307}, \doi{10.1007/978-3-540-28644-8\_19}.

\bibitemdeclare{inproceedings}{VriesKH10}
\bibitem{VriesKH10}
\bibinfo{author}{Edsko {de Vries}}, \bibinfo{author}{Vasileios Koutavas} \&
  \bibinfo{author}{Matthew Hennessy} (\bibinfo{year}{2010}):
  \emph{\bibinfo{title}{Communicating Transactions - (Extended Abstract)}}.
\newblock In: {\sl \bibinfo{booktitle}{CONCUR}}, {\sl \bibinfo{series}{LNCS}}
  \bibinfo{volume}{6269}, \bibinfo{publisher}{Springer}, pp.
  \bibinfo{pages}{569--583}, \doi{10.1007/978-3-642-15375-4\_39}.

\bibitemdeclare{inproceedings}{VKH10b}
\bibitem{VKH10b}
\bibinfo{author}{Edsko {de Vries}}, \bibinfo{author}{Vasileios Koutavas} \&
  \bibinfo{author}{Matthew Hennessy} (\bibinfo{year}{2010}):
  \emph{\bibinfo{title}{Liveness of Communicating Transactions - (Extended
  Abstract)}}.
\newblock In: {\sl \bibinfo{booktitle}{APLAS}}, {\sl \bibinfo{series}{LNCS}}
  \bibinfo{volume}{6461}, \bibinfo{publisher}{Springer}, pp.
  \bibinfo{pages}{392--407}, \doi{10.1007/978-3-642-17164-2\_27}.

\bibitemdeclare{inproceedings}{DY11}
\bibitem{DY11}
\bibinfo{author}{Pierre-Malo Deni{\'e}lou} \& \bibinfo{author}{Nobuko Yoshida}
  (\bibinfo{year}{2011}): \emph{\bibinfo{title}{{Dynamic Multirole Session
  Types}}}.
\newblock In: {\sl \bibinfo{booktitle}{POPL}}, \bibinfo{publisher}{ACM Press},
  pp. \bibinfo{pages}{435--446}, \doi{10.1145/1926385.1926435}.

\bibitemdeclare{inproceedings}{DGJPY16}
\bibitem{DGJPY16}
\bibinfo{author}{Mariangiola Dezani{-}Ciancaglini}, \bibinfo{author}{Silvia
  Ghilezan}, \bibinfo{author}{Svetlana Jaksic}, \bibinfo{author}{Jovanka
  Pantovic} \& \bibinfo{author}{Nobuko Yoshida} (\bibinfo{year}{2016}):
  \emph{\bibinfo{title}{Precise subtyping for synchronous multiparty
  sessions}}.
\newblock In: {\sl \bibinfo{booktitle}{PLACES}}, {\sl
  \bibinfo{series}{{EPTCS}}} \bibinfo{volume}{203}, pp.
  \bibinfo{pages}{29--43}, \doi{10.4204/EPTCS.203.3}.

\bibitemdeclare{article}{GH05}
\bibitem{GH05}
\bibinfo{author}{Simon Gay} \& \bibinfo{author}{Malcolm Hole}
  (\bibinfo{year}{2005}): \emph{\bibinfo{title}{{Subtyping for Session Types in
  the Pi Calculus}}}.
\newblock {\sl \bibinfo{journal}{Acta Informatica}}
  \bibinfo{volume}{42}(\bibinfo{number}{2/3}), pp. \bibinfo{pages}{191--225},
  \doi{10.1007/s00236-005-0177-z}.

\bibitemdeclare{inproceedings}{HVK98}
\bibitem{HVK98}
\bibinfo{author}{Kohei Honda}, \bibinfo{author}{Vasco~T. Vasconcelos} \&
  \bibinfo{author}{Makoto Kubo} (\bibinfo{year}{1998}):
  \emph{\bibinfo{title}{Language Primitives and Type Disciplines for Structured
  Communication-based Programming}}.
\newblock In: {\sl \bibinfo{booktitle}{ESOP}}, {\sl \bibinfo{series}{LNCS}}
  \bibinfo{volume}{1381}, \bibinfo{publisher}{Springer}, pp.
  \bibinfo{pages}{22--138}, \doi{10.1007/BFb0053567}.

\bibitemdeclare{inproceedings}{CHY08}
\bibitem{CHY08}
\bibinfo{author}{Kohei Honda}, \bibinfo{author}{Nobuko Yoshida} \&
  \bibinfo{author}{Marco Carbone} (\bibinfo{year}{2008}):
  \emph{\bibinfo{title}{Multiparty Asynchronous Session Types}}.
\newblock In: {\sl \bibinfo{booktitle}{POPL}}, \bibinfo{publisher}{ACM Press},
  pp. \bibinfo{pages}{273--284}, \doi{10.1145/1328897.1328472}.

\bibitemdeclare{inproceedings}{KSH14}
\bibitem{KSH14}
\bibinfo{author}{Vasileios Koutavas}, \bibinfo{author}{Carlo Spaccasassi} \&
  \bibinfo{author}{Matthew Hennessy} (\bibinfo{year}{2014}):
  \emph{\bibinfo{title}{Bisimulations for Communicating Transactions -
  (Extended Abstract)}}.
\newblock In: {\sl \bibinfo{booktitle}{FOSSACS}}, {\sl \bibinfo{series}{LNCS}}
  \bibinfo{volume}{8412}, \bibinfo{publisher}{Springer}, pp.
  \bibinfo{pages}{320--334}, \doi{10.1007/978-3-642-54830-7\_21}.

\bibitemdeclare{inproceedings}{LMSS11}
\bibitem{LMSS11}
\bibinfo{author}{Ivan Lanese}, \bibinfo{author}{Claudio~Antares Mezzina},
  \bibinfo{author}{Alan Schmitt} \& \bibinfo{author}{Jean-Bernard Stefani}
  (\bibinfo{year}{2011}): \emph{\bibinfo{title}{Controlling Reversibility in
  Higher-Order Pi}}.
\newblock In: {\sl \bibinfo{booktitle}{CONCUR}}, {\sl \bibinfo{series}{LNCS}}
  \bibinfo{volume}{6901}, \bibinfo{publisher}{Springer}, pp.
  \bibinfo{pages}{297--311}, \doi{10.1007/978-3-642-23217-6\_20}.

\bibitemdeclare{inproceedings}{LMS10}
\bibitem{LMS10}
\bibinfo{author}{Ivan Lanese}, \bibinfo{author}{Claudio~Antares Mezzina} \&
  \bibinfo{author}{Jean-Bernard Stefani} (\bibinfo{year}{2010}):
  \emph{\bibinfo{title}{Reversing Higher-Order Pi}}.
\newblock In: {\sl \bibinfo{booktitle}{CONCUR}}, {\sl \bibinfo{series}{LNCS}}
  \bibinfo{volume}{6269}, \bibinfo{publisher}{Springer}, pp.
  \bibinfo{pages}{478--493}, \doi{10.1007/978-3-642-15375-4\_33}.

\bibitemdeclare{inproceedings}{MP16}
\bibitem{MP16}
\bibinfo{author}{Claudio~A. Mezzina} \& \bibinfo{author}{Jorge~A. P\'erez}
  (\bibinfo{year}{2016}): \emph{\bibinfo{title}{Reversible Sessions Using
  Monitors}}.
\newblock In: {\sl \bibinfo{booktitle}{PLACES}}, {\sl
  \bibinfo{series}{{EPTCS}}} \bibinfo{volume}{211}, pp.
  \bibinfo{pages}{56--64}, \doi{10.4204/EPTCS.211.6}.

\bibitemdeclare{book}{M89}
\bibitem{M89}
\bibinfo{author}{Robin Milner} (\bibinfo{year}{1989}):
  \emph{\bibinfo{title}{Communication and concurrency}}.
\newblock \bibinfo{series}{PHI Series in computer science},
  \bibinfo{publisher}{Prentice Hall}.

\bibitemdeclare{inproceedings}{P11}
\bibitem{P11}
\bibinfo{author}{Luca Padovani} (\bibinfo{year}{2011}):
  \emph{\bibinfo{title}{Session Types = Intersection Types + Union Types}}.
\newblock In: {\sl \bibinfo{booktitle}{ITRS}}, {\sl
  \bibinfo{series}{EPTCS}}~\bibinfo{volume}{45}, pp. \bibinfo{pages}{71--89},
  \doi{10.4204/EPTCS.45.6}.

\bibitemdeclare{article}{PU07}
\bibitem{PU07}
\bibinfo{author}{Iain C.~C. Phillips} \& \bibinfo{author}{Irek Ulidowski}
  (\bibinfo{year}{2007}): \emph{\bibinfo{title}{Reversing algebraic process
  calculi}}.
\newblock {\sl \bibinfo{journal}{Journal of Logic and Algebraic Methods in
  Programming}} \bibinfo{volume}{73}(\bibinfo{number}{1-2}), pp.
  \bibinfo{pages}{70--96}, \doi{10.1016/j.jlap.2006.11.002}.

\bibitemdeclare{book}{pier02}
\bibitem{pier02}
\bibinfo{author}{Benjamin~C. Pierce} (\bibinfo{year}{2002}):
  \emph{\bibinfo{title}{Types and Programming Languages}}.
\newblock \bibinfo{publisher}{MIT Press}.

\bibitemdeclare{article}{TY15}
\bibitem{TY15}
\bibinfo{author}{Francesco Tiezzi} \& \bibinfo{author}{Nobuko Yoshida}
  (\bibinfo{year}{2015}): \emph{\bibinfo{title}{Reversible Session-Based
  Pi-Calculus}}.
\newblock {\sl \bibinfo{journal}{Journal of Logical and Algebraic Methods in
  Programming}} \bibinfo{volume}{84}(\bibinfo{number}{5}), pp.
  \bibinfo{pages}{684--707}, \doi{10.1016/j.jlamp.2015.03.004}.

\bibitemdeclare{inproceedings}{TY16}
\bibitem{TY16}
\bibinfo{author}{Francesco Tiezzi} \& \bibinfo{author}{Nobuko Yoshida}
  (\bibinfo{year}{2016}): \emph{\bibinfo{title}{Reversing Single Sessions}}.
\newblock In: {\sl \bibinfo{booktitle}{RC}}, {\sl \bibinfo{series}{LNCS}}
  \bibinfo{volume}{9720}, \bibinfo{publisher}{Springer}, pp.
  \bibinfo{pages}{52--69}, \doi{10.1007/978-3-319-40578-0\_4}.

\end{thebibliography}


\end{document}